%% file: main.tex
\documentclass[11pt,letterpaper]{article}
\input{macros}

\title{Matrix Samplers and Quantum-Proof Extractors: Their Correspondence and Almost-Optimal Constructions}
\author{
Kuan Cheng\thanks{CFCS, School of Computer Science, Peking University. ckkcdh@pku.edu.cn.}
\and
Ruiyang Wu\thanks{CFCS, School of Computer Science, Peking University. wuruiyang@stu.pku.edu.cn.}
}
\date{}
\hypersetup{pdftitle={Matrix Samplers and Quantum-Proof Extractors: Their Correspondence and Almost-Optimal Constructions},pdfauthor={Kuan Cheng and Ruiyang Wu},pdfsubject={Quantum-proof extraction, strong matrix sampling, and randomness complexity}}
\begin{document}
\maketitle

\input{abstract}
%\clearpage
\input{intro}

\input{preli}
\input{extractor_sampler_correspondence}
\input{quantum_extractor}
\input{seed_generation}
\input{high_entropy}
\input{matrix_samplers}
%\input{limitations}
\input{ack}

\input{refer}

\appendix
% The complete pre-pruning source is preserved in retained/pre_pruning/.
% It is intentionally not part of this compilation.
% Supplementary obstruction and expander-walk proofs are in
% retained/appendices_20261025/ and are not compiled.
\input{appendices/smooth_min_entropy}
\input{appendices/sampling_extraction_tools}

\end{document}

%% file: macros.tex
\usepackage[margin=1in]{geometry}
\usepackage{amsmath,amssymb,amsthm,mathtools,booktabs,array,needspace}
\usepackage{xcolor}
\usepackage[pagebackref=true,colorlinks,urlcolor=blue,linkcolor=purple,citecolor=blue,bookmarks,bookmarksopen,bookmarksnumbered]{hyperref}
\usepackage[capitalise,nameinlink]{cleveref}
\usepackage{titlesec}
\titleformat{\subsubsection}[runin]
{\normalfont\normalsize\bfseries}
{\thesubsubsection}{1em}{}
\numberwithin{equation}{section}
\newtheorem{theorem}{Theorem}[section]
\newtheorem{lemma}[theorem]{Lemma}
\newtheorem{corollary}[theorem]{Corollary}
\newtheorem*{claim}{Claim}
\newtheorem*{correspondencerestatement}{Theorem~\ref{thm:correspondence}}
\theoremstyle{definition}
\newtheorem{definition}[theorem]{Definition}

\newcommand{\E}{\mathbb E}
\newcommand{\F}{\mathbb F}
\newcommand{\Tr}{\operatorname{Tr}}
\newcommand{\Ext}{\operatorname{Ext}}
\newcommand{\Samp}{\operatorname{Samp}}
\newcommand{\hmin}{H_{\min}}
\newcommand{\pg}{p_{\mathrm{guess}}}
\newcommand{\cp}{\operatorname{cp}}
\newcommand{\dist}{\Delta}
\newcommand{\poly}{\operatorname{poly}}
\newcommand{\supp}{\operatorname{supp}}
\newcommand{\norm}[1]{\left\lVert #1\right\rVert}
\newcommand{\ket}[1]{|#1\rangle}
\newcommand{\bra}[1]{\langle #1|}

%% file: abstract.tex
\begin{abstract}
\normalsize
We establish an equivalence between matrix samplers and quantum-proof randomness extractors. The classical sampler--extractor equivalence is established by a seminal work of Zuckerman (Random Structures \& Algorithms, 1997).

Using this equivalence, we construct explicit strong matrix samplers with optimal randomness overhead and near-optimal sample complexity via constructing quantum-proof extractors. %Our sampler estimates the average of a matrix-valued function $f:\{0,1\}^m\to\mathbb C^{D\times D}$ with $\norm{f(x)}\leq1$. 
Specifically for matrix dimension $D$, accuracy $\varepsilon$ and failure probability $\delta\leq\varepsilon$, our matrix sampler uses
\[
 m+O_\alpha(\log(D/\delta))\text{ random bits and }
 O_\alpha\!\left(\varepsilon^{-2-\alpha}\log^{1+\alpha}(D/\delta)\right)\text{ samples},
\]
for any constant $\alpha>0$. 
%Compared with the matrix sampler of Xun and Zuckerman (CCC 2025), we remove the extra $\log\log D$ factor in the dimension-dependent randomness.
Compared with the matrix sampler of Xun and Zuckerman (CCC 2025), our construction has improvements in two aspects. Under optimal random seed length, we can replace the $\poly D$ dependence in the number of samples by $O(\log^{1+\alpha} D)$.
Under the same scale of sample size, we can remove their $\log \log D$ factor in the length of the random seed. 
Our underlying strong quantum-proof extractor has output length $n-O_\alpha(a+\log(1/\eta))$ and seed length at most $(1+\alpha)\log a+(2+\alpha)\log(1/\eta)+O_\alpha(1)$ for $n$-bit sources with conditional min-entropy $n-a$ and error $\eta$. 

We also obtain strong quantum-proof extractors for every constant entropy rate, for any error as small as $2^{-\Omega(n)}$, having optimal seed length $O(\log(n/\eta))$ and output length $\Omega(n)$, improving previous constructions.
\end{abstract}

%% file: intro.tex
\section{Introduction}
Sampling estimates the average of a function from its values at a small number of points. Independent samples give sharp concentration bounds, but require fresh random bits for every evaluation. A randomness-efficient sampler seeks the same accuracy and confidence using samples generated from a much shorter random string. A challenge is to reduce the randomness without substantially increasing the number of samples.

%Sampling for scalar-valued functions are closely connected to randomness extraction ~\cite{Zuc97}. This connection underlies samplers with nearly optimal parameters~\cite{RVW00}. Matrix sampling and extraction in the presence of quantum side information are two natural extensions of the classical theory. They motivate structural questions about the corresponding notions of sampling and extraction, as well as quantitative questions about how efficiently these objects can be constructed.

\paragraph{From scalar to matrix sampler.}
For $0<\varepsilon,\delta\leq1/2$, consider a function $f:\{0,1\}^m\to[0,1]$ whose mean we wish to estimate to within additive error $\varepsilon$, by a randomized procedure with failure probability at most $\delta$. The Chernoff bound shows that $O(\varepsilon^{-2}\log(1/\delta))$ independent uniform samples suffice. Generating these samples independently, however, costs $m$ random bits per sample. An averaging sampler can replace independent samples with correlated samples whose empirical average satisfies the same guarantee~\cite{BellareR94,Goldreich}.

\Needspace{9\baselineskip}
\begin{definition}[Averaging samplers~\cite{BellareR94}]\label{def:averaging-sampler}
Let $\Samp:\{0,1\}^r\times[t]\to\{0,1\}^m$, where $[t]=\{1,\ldots,t\}$. Let $R$ be uniform on $\{0,1\}^r$, and let $Z$ be uniform on $\{0,1\}^m$. The map $\Samp$ is a $(\delta,\varepsilon)$ averaging sampler if, for every $f:\{0,1\}^m\to[0,1]$,
\[
 \Pr_R\!\left[\left|\frac1t\sum_{i=1}^t f(\Samp(R,i))-\E_Z f(Z)\right|>\varepsilon\right]\leq\delta.
\]
\end{definition}
The main parameters here are the randomness complexity $r$ and the sample complexity $t$. Classical approaches include limited independence, expander walks, and iterated sampling~\cite{CG89,AKS87,IZ89,Gill98,Healy08,BellareR94}. Goldreich's survey~\cite{Goldreich} gives an account of these methods and distinguishes averaging samplers from general samplers, such as median-of-averages constructions~\cite{BGG93}, that may use a different estimator; see also Canetti, Even, and Goldreich~\cite{CEG95} for lower bounds. Zuckerman~\cite{Zuc97} established a correspondence between averaging samplers and extractors. Applying this correspondence to their high-min-entropy extractors, Reingold, Vadhan, and Wigderson~\cite{RVW00} obtained averaging samplers with nearly optimal randomness complexity. More recently, for $\delta<\varepsilon$, Xun and Zuckerman~\cite{XZ} achieved $m+O(\log(1/\delta))$ random bits with sample count $O_\alpha((\varepsilon^{-2}\log(1/\delta))^{1+\alpha})$ for every fixed $\alpha>0$.

There is a matrix analogue of the sampling task for estimating scalar functions. Specifically the matrix analogue asks us to estimate the mean of a function $f:\{0,1\}^m\to\mathbb C^{D\times D}$, where $\norm{f(z)}\leq1$, and to measure the error in operator norm, $\norm A=\max_{\|v\|_2=1}\|Av\|_2$. Thus the approximation must be accurate in every direction simultaneously. Matrix concentration inequalities, applied through Hermitian dilation for general complex matrices, show that $O(\varepsilon^{-2}\log(2D/\delta))$ independent samples suffice, so the dimension enters only logarithmically~\cite{AW,Tropp12}. Such bounds are basic tools in spectral graph theory, numerical linear algebra, and quantum information. Their derandomization motivates the study of matrix (averaging) samplers~\cite{WX08,GLSS,XZ}.

\Needspace{14\baselineskip}
\begin{definition}[Matrix (Averaging) samplers~\cite{XZ}]\label{def:sampler}\label{def:strong-sampler}
Let $\Samp:\{0,1\}^r\times[t]\to\{0,1\}^m$. Let $R$ be uniform on $\{0,1\}^r$, and let $Z$ be uniform on $\{0,1\}^m$.

The map $\Samp$ is a $D$-dimensional $(\delta,\varepsilon)$ matrix sampler if, for every $f:\{0,1\}^m\to\mathbb C^{D\times D}$ satisfying $\norm{f(z)}\leq1$ for every $z$,
\[
 \Pr_R\!\left[\left\lVert\frac1t\sum_{i=1}^t f(\Samp(R,i))-\E_Z f(Z)\right\rVert>\varepsilon\right]\leq\delta.
\]
The map $\Samp$ is a \emph{strong} $D$-dimensional $(\delta,\varepsilon)$ matrix sampler if, for every sequence $f_1,\ldots,f_t:\{0,1\}^m\to\mathbb C^{D\times D}$ satisfying $\norm{f_i(z)}\leq1$ for every $i,z$,
\[
 \Pr_R\!\left[\left\lVert\frac1t\sum_{i=1}^t\bigl(f_i(\Samp(R,i))-\E_Z f_i(Z)\bigr)\right\rVert>\varepsilon\right]\leq\delta.
\]
We  may use the term \emph{ordinary matrix sampler} to distinguish matrix sampler from strong matrix sampler.
\end{definition}
We call $\Samp$ explicit if the map $R\mapsto(\Samp(R,1),\ldots,\Samp(R,t))$ is computable in time polynomial in $r+mt$. The strong version allows the test to depend on the sample position. The challenge is to satisfy this requirement while keeping both the randomness complexity and the sample complexity small.

The matrix problem has developed along two complementary lines. Garg, Lee, Song, and Srivastava~\cite{GLSS} established a matrix Chernoff bound for expander walks, resolving the conjecture arising from the work of Wigderson and Xiao~\cite{WX08}. Subsequent extensions cover finite Markov chains~\cite{QWLPT20} and time-dependent matrix observables on general state spaces~\cite{NSW}. Song's recent work~\cite{Song26} sharpens the Hoeffding and Bernstein constants, including in the expander-walk setting. For a stationary walk on a fixed-degree expander with constant absolute spectral gap, these bounds give sample count $t=O(\varepsilon^{-2}\log(2D/\delta))$, matching the independent-sampling bound. Generating the standard walk uses $m+\Theta(t)$ random bits, so this construction retains an $\varepsilon^{-2}$ factor in its randomness overhead.

The other line seeks to reduce this overhead directly. Applying a scalar sampler entry by entry gives good randomness complexity, but introduces a polynomial dependence on $D$ in the sample count~\cite{XZ}. In the regime $\delta\leq\varepsilon$, Xun and Zuckerman~\cite{XZ} substantially improved the tradeoff through sampler composition: their constructions offer either nearly optimal sample complexity with an additional dimension-dependent randomness cost, or randomness overhead $O(\log(D/\delta))$ with a polynomial dependence on $D$ in the sample count. Their constructions also satisfy the strong sampling requirement~\cite[Remark~9]{XZ}. 
%Beyond these parameter tradeoffs, the classic extractor--sampler correspondence \cite{Zuc97} suggests a structural question about matrix sampling:
%\begin{quote}
%\emph{What type of randomness extractor corresponds to a matrix sampler?}
%\end{quote}

\paragraph{Randomness extraction and quantum side information.}
As shown by Zuckerman \cite{Zuc97}, there is a correspondence between averaging samplers and randomness extractors. 
A seeded extractor converts a weak random source into nearly uniform bits with the help of a short independent random seed. We measure the randomness of a source $X$ by its min-entropy $H_\infty(X)=-\log_2\max_x\Pr[X=x]$. Write $U_j$ for the uniform distribution on $j$ bits and $\dist(P,Q)=\frac12\sum_z|P(z)-Q(z)|$ for total variation distance.

\begin{definition}[Seeded extractors~\cite{NZ96}]\label{def:classical-extractor}
A function $\Ext:\{0,1\}^n\times\{0,1\}^d\to\{0,1\}^m$ is a $(k,\eta)$ extractor if, for every distribution $X$ on $\{0,1\}^n$ with $H_\infty(X)\geq k$ and an independent uniform seed $Y\sim U_d$,
\[
 \dist(\Ext(X,Y),U_m)\leq\eta.
\]
The map $\Ext$ is \emph{strong} if, for every $X$ with $H_\infty(X)\geq k$,
\[
 \dist((\Ext(X,Y),Y),U_{m+d})\leq\eta.
\]
\end{definition}

The classical goal is to maximize the output while minimizing the seed length, subject to the lower bounds of Radhakrishnan and Ta-Shma~\cite{RT00}. Among the strongest general constructions, Guruswami, Umans, and Vadhan~\cite{GUV09} gave explicit extractors with seed length $d\leq\log n+O(\log(k/\eta))$ and output length $m\geq(1-\gamma)k$ for every fixed $0<\gamma<1$, throughout the full range of error parameters. Dvir, Kopparty, Saraf, and Sudan~\cite{DKSS13} and Ta-Shma and Umans~\cite{TU12} achieved output length $m\geq k-O(k/\log^c n)$ with seed length $d=O(\log n)$, for every fixed $c>0$ and $0<\beta<1$, when $\eta\geq2^{-(\log n)^\beta}$.\footnote{The smaller-error range for the construction of \cite{DKSS13} follows from the refinement noted by Ta-Shma and Umans~\cite[Section~1]{TU12}; the original statement in~\cite{DKSS13} gives inverse-polylogarithmic error.} The latter two constructions obtain their improved output through multiplicity-based mergers and a refinement of the Parvaresh--Vardy-code condenser, respectively.

A large body of related techniques have been used in or developed for constructing extractors. Hashing provides basic extraction procedures~\cite{BBR88,ILL89,IZ89,GW97}. Block-source methods exploit conditional min-entropy guarantees for successive blocks of a weak random source~\cite{CG88,NZ96,SZ99}. Repeated extraction increases the amount of randomness recovered from a source~\cite{WZ99}. Pseudorandom-generator techniques underlie Trevisan's construction and its refinements, combining the Nisan--Wigderson generator with error-correcting codes and combinatorial designs~\cite{Tre01,RRV02}.  Repeated condensing combines condensers and extractors to improve the seed-length and output-length tradeoff~\cite{RSW00,LRVW03}. Error-reduction techniques improve the accuracy of extraction and control error accumulation in composed constructions~\cite{RRV99,RRV02,LRVW03}. Somewhere-extractor and merger techniques organize extraction by first producing a somewhere-random source and then merging its blocks~\cite{NT99,DW11,DKSS13}. Advanced Condenser constructions increase the entropy rate range while reducing entropy loss~\cite{TUZ07,TU06,GUV09,TU12,KT22}. Extractor constructions suitable for low-depth circuits realizations~\cite{GVW15,CL18,CW24}. The surveys of Shaltiel~\cite{Shaltiel02} and Vadhan~\cite{Vadhan12} provide background on classical extractors and their connections to pseudorandomness.

For sampling, the relevant connection is particularly simple to describe: for a fixed source string $x$, enumerate all seeds $y$ and use the outputs $\Ext(x,y)$ as the sample points. A uniformly chosen source string becomes the sampler's random input, and a $d$-bit extractor seed gives $2^d$ samples. Zuckerman's correspondence~\cite{Zuc97} makes this relationship precise in both directions. Additional connections to list-decodable codes and entropy sampling appear in~\cite{TZ04,NZ96,Vad04}.
As matrix sampler is a natural extension of samplers, the classic extractor--sampler correspondence \cite{Zuc97} suggests a structural question about matrix sampling:
\begin{quote}
\emph{What type of randomness extractor corresponds to a matrix sampler?}
\end{quote}

On the other hand, quantum side information changes the extraction problem. In randomness extraction, the extracted string must be nearly uniform even to an adversary who has partial information about the source~\cite{BBR88,BBCM95}. If this information is stored in a quantum system, security must hold against measurements of that system~\cite{Ren05,LHL}. Although the extractor itself is a classical function, classical security does not in general imply quantum security~\cite{GKKRW08}.

A classic random variable $X$ together with a quantum register $E$, can be described as a classical--quantum (cq) state represented by $\rho_{XE}$ which is the density matrix for the joint distribution $(X, E)$. That is, $(X, E)$ is regarded as a quantum mixed state which is a convex combination of mixed states corresponding to $(x, E|_{X=x} ), x\in \supp(X)$, where the convex combination is according to the distribution of $X$. 
There are two relevant models for quantum-proof extractors. The bounded-storage model places a lower bound on $H_\infty(X)$ and an upper bound on $\dim E$, where $\dim E$ denotes the dimension of the Hilbert space associated with the quantum register $E$.
The more general conditional min-entropy model directly bounds the uncertainty remaining about $X$ given $E$: $\hmin(X\mid E)=-\log_2\pg(X\mid E)$, where $\pg(X\mid E)$ is the optimal probability of guessing $X$ by measuring $E$~\cite{KRS}. We measure error in trace distance, $\dist(\rho,\sigma)=\frac12\norm{\rho-\sigma}_1$, and write $\tau_j=2^{-j}I_{2^j}$ for the uniform state on $j$ bits.

\begin{definition}[Quantum-proof extractors~\cite{DPVR}]\label{def:extractor}
A function $\Ext:\{0,1\}^n\times\{0,1\}^d\to\{0,1\}^m$ is a \emph{weak quantum-proof $(k,\eta)$ extractor} if, for every cq state $\rho_{XE}$ with $X\in\{0,1\}^n$ and $\hmin(X\mid E)\geq k$, and an independent uniform seed $Y\sim U_d$,
\begin{equation}\label{eq:weak-extraction}
 \dist(\rho_{\Ext(X,Y)E},\tau_m\otimes\rho_E)\leq\eta.
\end{equation}
The map $\Ext$ is \emph{strong} if, for every cq state $\rho_{XE}$ with $\hmin(X\mid E)\geq k$,
\begin{equation}\label{eq:strong-extraction}
 \dist(\rho_{\Ext(X,Y)YE},\tau_m\otimes\tau_d\otimes\rho_E)\leq\eta.
\end{equation}
We call $\Ext$ a \emph{weak $b$-qubit-proof $(k,\eta)$ extractor} if~\eqref{eq:weak-extraction} holds for every cq state $\rho_{XE}$ with $H_\infty(X)\geq k$ and $\dim E\leq2^b$. The \emph{strong} version requires~\eqref{eq:strong-extraction} instead~\cite[Definition~2.3]{DeVidick10}.
\end{definition}

Several classical constructions are known to remain secure in this setting. K{\"o}nig and Terhal~\cite{KT08} established quantum security for strong one-bit extractors, with a loss in parameters. Ta-Shma~\cite{TaShma11} and De and Vidick~\cite{DeVidick10} developed short-seed constructions against bounded quantum storage based on Trevisan's approach. De, Portmann, Vidick, and Renner~\cite{DPVR} extended this approach to conditional min-entropy. At inverse-polynomial error, their constructions give entropy loss $O(\log n)$ with a polylogarithmic seed, or a logarithmic seed with polynomially smaller output for polynomial min-entropy. Ben-Aroya and Ta-Shma~\cite{BATS} combined lossless condensing with high-entropy quantum extraction to obtain short-seed constructions in the bounded-storage model, and also studied conditional entropy rates above one half. Berta, Fawzi, and Scholz~\cite{BFS} provided an operator-space perspective, formulating classical and quantum-proof extraction through bounded and completely bounded norm conditions.

In light of the classical correspondence \cite{Zuc97} between averaging samplers and extractors, it is also natural to ask for a sampling interpretation of quantum-proof extractors. Such an interpretation would need to reflect the adversary's access to side information and, for strong extractors, to the public seed. This motivates the structural question:
\begin{quote}
\emph{What type of sampler corresponds to a quantum-proof extractor?}
\end{quote}

\paragraph{Improving constructions.}
Alongside these structural questions, the parameter tradeoffs described above motivate the search for better explicit constructions. For matrix samplers, the challenge is to reduce the randomness overhead while keeping the sample complexity close to that of independent sampling, with mild dependence on the matrix dimension. For quantum-proof extractors, the challenge is to shorten the seed while retaining a large output over a wide range of entropy and error parameters. These are simultaneous requirements: improving one resource at the expense of another need not improve the overall tradeoff. This leads to a third question:
\begin{quote}
\emph{How to improve explicit constructions of matrix samplers and quantum-proof extractors?}
\end{quote}

\input{our_results}

\input{technical_overview}

%% file: our_results.tex
\Needspace{12\baselineskip}
\subsection{Our Results}\label{sec:results}
\paragraph{The correspondence} 
Our first result establishes a correspondence between quantum-proof extractors and matrix samplers.

\begin{theorem}[Quantum extractor--matrix sampler correspondence]\label{thm:correspondence}\label{thm:witness}\label{prop:converse}\label{prop:strong-bridge}
\input{statements/correspondence}
\end{theorem}

Every strong (resp.\ weak) quantum-proof $(k,\eta)$ extractor is a strong (resp.\ weak) $b$-qubit-proof $(k+b,\eta)$ extractor, by~\eqref{eq:chain}. Therefore, a strong (resp.\ weak) quantum-proof $(n-\log(2D/\delta),\varepsilon/2)$ extractor $F:\{0,1\}^n\times\{0,1\}^d\to\{0,1\}^m$ satisfies the hypothesis of Theorem~\ref{thm:correspondence}, part~(1). %Identifying $\{0,1\}^d$ with $[2^d]$, the same map is a strong (resp.\ ordinary) $D$-dimensional $(\delta,\varepsilon)$ matrix sampler $\Samp:\{0,1\}^n\times[2^d]\to\{0,1\}^m$.

\Needspace{8\baselineskip}
\paragraph{Strong matrix samplers} For matrix samplers, we attain the following explicit construction which has optimal random seed length and simultaneously has the sample complexity being at most the $(1+\alpha)$-th power of the optimal bound, for any constant $\alpha$. %(For the lower bound see Lemma~\ref{prop:strong-randomness-lower}).
\begin{theorem}[Strong matrix sampling]\label{thm:extractor-sampling}
For every constant $\alpha>0$, for every $m,D\geq1$ and $0<\delta\leq\varepsilon\leq1/2$, there is an explicit strong $D$-dimensional $(\delta,\varepsilon)$ matrix sampler $\Samp:\{0,1\}^r\times[t]\to\{0,1\}^m$ such that,
\begin{equation}\label{eq:main-sampler}
 r=m+O_\alpha(\Lambda),\qquad
 t=O_\alpha\!\left(\frac{\Lambda^{1+\alpha}\varepsilon^{-2-\alpha}}{\log(\Lambda/\varepsilon)}\right),
\end{equation}
where $\Lambda=\log_2(2D/\delta)$.
For every $x\in\{0,1\}^r$ and $i\in[t]$, $\Samp(x,i)$ is computable from $x$ and $i$ in time polynomial in $m+\Lambda+\log(1/\varepsilon)$.
\end{theorem}
Compared with Xun and Zuckerman~\cite[Theorem~3 and Remark~9]{XZ}, Theorem~\ref{thm:extractor-sampling} removes the extra $\log \log D$ factor from the dimension-dependent term in the random seed length. Comparing with their second construction~\cite[Theorem~4]{XZ}, our construction improves the sample complexity dependence on $D$ from a polynomial in $D$ to $O(\log (D/\delta))^{1+\alpha}$. See Table~\ref{tab:matrix-samplers} for detailed parameter comparison with previous works.

\input{tables/matrix_samplers}

\paragraph{The Quantum-proof extractor for constructing our matrix sampler}
Our matrix sampler is attained via constructing the following quantum-proof extractor and deploying our correspondence between quantum-proof extractors and matrix samplers.
\begin{theorem}\label{thm:deficiency}
%Fix a constant $\alpha>0$. For every integer $n\geq1$, entropy threshold $0\leq k\leq n$, and error $0<\eta\leq1/2$, there is an explicit strong quantum-proof $(k,\eta)$ extractor $\Ext:\{0,1\}^n\times\{0,1\}^d\to\{0,1\}^m$ such that, writing $a=n-k$,
For every constant $\alpha>0$, for every integer $n\geq1$, entropy $0\leq k\leq n$ and every error $0<\eta\leq1/2$, there is an explicit strong quantum-proof $(k,\eta)$ extractor $\Ext:\{0,1\}^n\times\{0,1\}^d\to\{0,1\}^m$ such that, writing $a=n-k$,
\begin{align}
 d&\leq(1+\alpha)\log a+(2+\alpha)\log(1/\eta)+O_\alpha(1),\label{eq:deficiency-seed}\\
 m&\geq n-O_\alpha(a+\log(1/\eta)).\label{eq:deficiency-output}
\end{align}
The map $\Ext$ is computable in polynomial time.
\end{theorem}
Notice that this theorem is used for constructing matrix samplers and the range of $k$ is only meaningful when $\frac{k}{n}$ is at least a constant close to 1, since otherwise the output length will be $0$.

%The strong averaging sampler of Xun and Zuckerman, combined with the classical correspondence~\cite[Theorem~1 and Lemma~33]{XZ}, gives the seed bound~\eqref{eq:deficiency-seed} and output length $n-O_\alpha(a+\log(1/\eta))$ for classical sources.\footnote{Take sampler accuracy $\eta/2$ and failure probability $(\eta/2)2^{-a}$. The additive randomness bound gives the output length. The sample bound gives~\eqref{eq:deficiency-seed} after choosing a sufficiently small exponent in their construction.} Theorem~\ref{thm:deficiency} achieves these parameters against quantum side information.

\paragraph{Quantum-proof extractor for any constant entropy rate}
We also obtain optimal up to constant factor explicit quantum-proof extractors for any constant entropy rate and any error from constant fraction to exponentially small in $n$.
%At any fixed positive entropy rate, we obtain linear output with seed length logarithmic in the source length and inverse error.

\begin{theorem}\label{thm:constant-rate}\label{cor:highentropy}
For any constant $0<\mu<1$, there is a constant $c_\mu>0$ such that, for all sufficiently large $n$, every $\mu n\leq k\leq n$, and every $0<\eta\leq1/2$ satisfying $\log(n/\eta)\leq c_\mu n$, there is an explicit strong quantum-proof $(k,\eta)$ extractor $\Ext:\{0,1\}^n\times\{0,1\}^d\to\{0,1\}^m$ such that
\begin{equation}\label{eq:highentropy-main}
 d=O_\mu(\log(n/\eta)),\qquad m=\Omega_\mu(n).
\end{equation}
The map $\Ext$ is computable in polynomial time.
\end{theorem}

Notice that our error can be as small as $2^{-c'_\mu n}$ for a small constant $c'_\mu>0$. 
For the output length, one can use a standard technique to attain an output length $(1-\theta) k $ for any constant $\theta$ (Corollary~\ref{cor:entropy-fraction}), with a constant factor blow-up for the seed length. 
The extractor of Ben-Aroya and Ta-Shma~\cite[Theorem~1.3]{BATS} achieves the same seed and output lengths as ours, but their extractor requires the conditional min-entropy $k \ge (\tfrac12+\beta)n$ for any constant $\beta > 0$, while our construction can work for all constant entropy rate.
Our work can also handle a larger range of error parameters than \cite{BATS}.
For details of comparison with previous works, see Table~\ref{tab:extractors}.

\input{tables/extractors}

%% file: statements/correspondence.tex
Let $F:\{0,1\}^n\times[t]\to\{0,1\}^m$. Fix $D\geq1$ and $0<\delta,\varepsilon<1$. Then:
\begin{enumerate}
\item If $F$ is a strong (resp.\ weak) $\log(2D)$-qubit-proof $(n-\log(1/\delta),\varepsilon/2)$ extractor, then $F$ is a strong (resp.\ ordinary) $D$-dimensional $(\delta,\varepsilon)$ matrix sampler.
\item Suppose $F$ is a strong (resp.\ ordinary) $D$-dimensional $(\delta,\varepsilon)$ matrix sampler. For every $0\leq a\leq n$, the map $F$ is a strong (resp.\ weak) $(\log D)$-qubit-proof $(n-a,\varepsilon/2+2^a\delta)$ extractor.
\end{enumerate}

%% file: tables/matrix_samplers.tex
\begin{table}[!htb]
\centering
\small
\setlength{\tabcolsep}{2pt}
\renewcommand{\arraystretch}{1.3}
\makebox[\linewidth][c]{%
\begin{tabular}{@{}llll@{}}
\toprule
Reference & Method & Random bits $r$ & Sample complexity $t$ \\
\midrule
Lemma~\ref{prop:strong-randomness-lower} & Lower bound (strong) & $m+\log(D/\delta)-1$ & $\Omega(\varepsilon^{-2}\log(D/\delta))$ \\
\cite[Prop.~4]{XZ} & Existence & $m+2\log(D/\delta)+\log\log(D/\varepsilon)+O(1)$ & $O(\varepsilon^{-2}\log(D/\delta))$ \\
\midrule
\cite{AW,Tropp12} & Full independence & $O(m\varepsilon^{-2}\log(D/\delta))$ & $O(\varepsilon^{-2}\log(D/\delta))$ \\
%\cite{XZ}$^{\dagger}$ & Union bound over entries & $m+O(\log(D/\delta))$ & %$O((D^2\varepsilon^{-2}\log(D/\delta))^{1+\gamma})$ \\
\cite{GLSS,NSW} & Expander walks & $m+O(\varepsilon^{-2}\log(D/\delta))$ & $O(\varepsilon^{-2}\log(D/\delta))$ \\
\cite[Thm.~3]{XZ} & Iterated composition & $m+O(\log(1/\delta)+\log(D/\varepsilon)\log\log D)$ & $O((\varepsilon^{-2}\log(D/\delta))^{1+\gamma})$ \\
\cite[Thm.~4]{XZ} & Sampler composition & $m+O(\log(D/\delta))$ & $O(D^\gamma\varepsilon^{-2-\gamma}\log^{1+\gamma}(1/\delta))$ \\
\midrule
Theorem~\ref{thm:extractor-sampling} & Quantum-proof extraction & $m+O(\log(D/\delta))$ & $O\!\left(\dfrac{\varepsilon^{-2-\alpha}\log^{1+\alpha}(D/\delta)}{\log(\varepsilon^{-1}\log(D/\delta))}\right)$ \\
\addlinespace[2pt]
Corollary~\ref{cor:quadratic-accuracy} & Quantum-proof extraction & $m+O(\log(D/\delta)\log(\varepsilon^{-1}\log(D/\delta)))$ & $O(\varepsilon^{-2}\log(D/\delta))$ \\
\bottomrule
\end{tabular}%
}
\caption{Matrix samplers on $\{0,1\}^m$, for $0<\delta\leq\varepsilon\leq1/2$. Hidden constants may depend on fixed $\alpha,\gamma>0$, and logarithmic factors are at least one. %Both lower bounds apply to strong samplers. %Both constructions of ours are strong. 
%$^{\dagger}$The entrywise bound combines Theorem~1 and Lemma~17 of~\cite{XZ}.
}
\label{tab:matrix-samplers}
\end{table}

%% file: tables/extractors.tex
\begin{table}[!htb]
\centering
\normalsize
\setlength{\tabcolsep}{2pt}
\renewcommand{\arraystretch}{1.2}
\begin{tabular*}{\linewidth}{@{\extracolsep{\fill}}lllll@{}}
\toprule
Reference & Entropy & Seed length $d$ & Output length $m$ & Error $\eta$ \\
\midrule
\multicolumn{5}{@{}l}{\textit{Bounded quantum storage: $H_\infty(X)\geq k$, $\dim E\leq2^b$}}\\[2pt]
\cite{TaShma11} & $k=n$, $b=n^\beta$ & $O(\log n)$ & $\Omega(n^{(1-\beta)/15-\gamma}/\log n)$ & $n^{-\gamma}$ \\
\cite[Thm.~1.1]{DeVidick10} & $k\geq n^\nu$ & $O(\log n)$ & $\Omega((k-b)/k^\gamma)$ & $n^{-c}$ \\
\cite[Thm.~1.2]{DeVidick10} & $k=\alpha n$ & $O(\log^4 n)$ & $\Omega(k-b)$ & $n^{-c}$ \\
\cite[Thm.~1.2]{BATS} & Any $k$, $b=\beta k$ & $O_\beta(\log(n/\eta))$ & $\Omega_\beta(k)$ & $\geq2^{-k^\beta}$ \\
\midrule
\multicolumn{5}{@{}l}{\textit{Quantum-proof extraction: $\hmin(X\mid E)\geq k$}}\\[2pt]
\cite{KT08,DPVR} & $k=\Omega(\log(1/\eta))$ & $O(\log(n/\eta))$ & $1$ & Any $\eta$ \\
\cite{DPVR} & Any $k$ & $O(\log^2(n/\eta)\log m)$ & $\geq k-4\log(1/\eta)-O(1)$ & Any $\eta$ \\
\cite[Cor.~5.6]{DPVR} & $k=n^\nu$ & $O(\log n)$ & $\Omega(n^{\nu-\gamma})$ & $n^{-c}$ \\
\cite{BATS} & $k=(\tfrac12+\beta)n$ & $O_\beta(\log(n/\eta))$ & $\Omega_\beta(n)$ & $\geq2^{-n^{\frac12-\beta}}$ \\

Theorem~\ref{thm:constant-rate} & $k=\mu n$ & $O_\mu(\log(n/\eta))$ & $\Omega_\mu(n)$ & $\geq2^{-c'_\mu n}$ \\
Theorem~\ref{thm:deficiency} & $k=n-a$ & $O_\alpha(\log(a/\eta))$ & $\geq n-O_\alpha(a+\log(1/\eta))$ & Any $\eta$ \\
\bottomrule
\end{tabular*}
\caption{Strong extractors secure against quantum side information, for sources on $\{0,1\}^n$ and $0<\eta\leq1/2$. Except $n, k, \eta, b, d, m, a$, all parameters are constants.  In~\cite{TaShma11}, $0<\beta<1$ and $0<\gamma<(1-\beta)/15$.
In~\cite{DeVidick10}, $\gamma,c,\nu>0$, $0<\alpha,\zeta<1$,
and $b\leq(1-\zeta)k$.
For~\cite[Corollary~5.6]{DPVR}, $0<\gamma<\nu\leq1$ and $c>0$. In~\cite{BATS}, $0<\beta<1/2$. For our results, $0<\mu<1$, $\alpha>0$, and $1\leq a\leq n$; $c'_\mu>0$ is sufficiently small and $n$ sufficiently large. 
%The one-bit row combines~\cite{KT08} with~\cite[Proposition~5.2]{DPVR}.
%$^\dagger$This rows report the parameters stated in~\cite[Theorems~1.2--1.3]{BATS}. Their small-error ranges use Theorem~4.3 of that paper, attributed to ~\cite{DPVR}; however,~\cite[Corollary~5.6]{DPVR} assumes inverse-polynomial error and does not itself justify these extensions.
}
\label{tab:extractors}
\end{table}

%% file: technical_overview.tex
\subsection{Technical Overview}\label{sec:overview}
\paragraph{The quantum extractor--matrix sampler correspondence.}
Let $F:\{0,1\}^n\times\{0,1\}^d\to\{0,1\}^m$ be a strong quantum-proof $(n-\Lambda,\varepsilon/2)$ extractor, where $\Lambda=\log(2D/\delta)$. 
We directly regard it as a matrix sampler with input randomness length $n$ and number of samples being $2^d$.
For every $y \in \{0, 1\}^d$, let $f_y:\{0,1\}^m\to\mathbb C^{D\times D}$ be an arbitrary matrix-valued function with $\norm{f_y(z)}\leq1$.
Consider the sampling error matrix
\[
 A_x=\E_y f_y(F(x,y))-\E_{y,z}f_y(z),
\]
where $y $ and $z $ are independent and uniform. Suppose the bad set $\mathcal B=\{x:\norm{A_x}>\varepsilon\}$ has size greater than $\delta2^n$. We will use these bad inputs $x$ to construct a source and quantum side information that violate the strong quantum-proof extractor guarantee.

Write $\mathcal H(A)=\bigl(\begin{smallmatrix}0&A\\A^*&0\end{smallmatrix}\bigr)$ for the Hermitian dilation of $A$. For each $x\in\mathcal B$, choose a unit eigenvector $v_x$ of $\mathcal H(A_x)$ with eigenvalue $\norm{A_x}$. Choose $X$ uniformly from $\mathcal B$. Conditioned on $X=x$, put the register $E$ in the state $\ket{v_x}\bra{v_x}$. This defines a cq state $\rho_{XE}$ with $\dim E=2D$, and hence
\[
 \hmin(X\mid E)_\rho\geq\log|\mathcal B|-\log(2D)>n-\Lambda,
\]
so the source meets the extractor's entropy requirement.

\iffalse
For each seed $i\in[t]$ and output $z\in\{0,1\}^m$, define
$M_{i,z}=(I_{2D}+\mathcal H(f_i(z)))/2$.
Since $0\preceq M_{i,z}\preceq I_{2D}$, the operator
\(
 M=\sum_{z\in\{0,1\}^m}\sum_{i=1}^t
 \ket{z}\bra{z}\otimes\ket{i}\bra{i}\otimes M_{i,z}
\)
defines a POVM $\{M,I-M\}$.
We accept if the outcome corresponds to $M$.
We show that the acceptance probabilities of measuring
$\rho_{F(X,I)IE}$ and $\tau_m\otimes\rho_{IE}$ differ by
\(
 \varepsilon/2,
\)
contradicting the extractor guarantee.
\fi

For each seed $y\in\{0,1\}^d$ and output $z\in\{0,1\}^m$, define
$M_{y,z}=(I_{2D}+\mathcal H(f_y(z)))/2$.
Since $0\preceq M_{y,z}\preceq I_{2D}$, the operator
\[
 M=\sum_{z\in\{0,1\}^m}\sum_{y\in\{0,1\}^d}
 \ket{z}\bra{z}\otimes\ket{y}\bra{y}\otimes M_{y,z}
\]
defines a POVM $\{M,I-M\}$.
We accept if the outcome corresponds to $M$.
Let $Y$ be an independent uniform seed on $\{0,1\}^d$.
The acceptance probabilities of this measurement on
$\rho_{F(X,Y)YE}$ and $\tau_m\otimes\rho_{YE}$
differ by more than $\varepsilon/2$,
contradicting the extractor guarantee.

Conversely, suppose $F:\{0, 1\}^n \times \{0, 1\}^d \rightarrow \{0, 1\}^m$ is a strong $D$-dimensional $(\delta,\varepsilon)$ matrix sampler. 
We directly view $F$ as a quantum-proof extractor.
Consider an arbitrary cq state $\rho_{XE}$ with $H_\infty(X)_\rho\geq n-a$ and $\dim E\leq D$.  Consider any binary measurement $\{M,I_{2^{d+m}D}-M\}$ on the
public seed $Y$, the output $Z=F(X,Y)$, and the quantum register $E$.
Since $Y$ and $Z$ are classical, we may assume without loss of
generality that the measurement operator $M$ has the form
\[
M=\sum_{y\in\{0,1\}^d}\sum_{z\in\{0,1\}^m}
|y\rangle\langle y|\otimes |z\rangle\langle z|\otimes M_{y,z},
\]
where $0\preceq M_{y,z}\preceq I_D$ for every $y,z$.
 Define the matrix tests $f_y(z)=2M_{y,z}-I$, which satisfy $\norm{f_y(z)}\leq1$. We show that by definitions of matrix samplers and quantum-proof extractors, the sampler guarantee for these tests can bound the distinguishing advantage of the measurement by $\varepsilon/2$ on all but at most $\delta $ fraction of $X$. Since for each $x$,  $\Pr[X= x] \le 2^{a-n}$, these bad input $x$ have total probability at most $2^a\delta$ under $X$. The advantage is at most one for these bad $x$, so its average is at most $\varepsilon/2+2^a\delta$. This proves the reverse implication of Theorem~\ref{thm:correspondence}.

\paragraph{A quantum-proof extractor with linear output.}
Fix a constant $0<\mu\leq1$ and put $h=\log(n/\eta)$. Assume $h\leq c_\mu n$ for a sufficiently small constant $c_\mu>0$. 
In our first stage, we will use the quantum left-over hash lemma \cite{Ren05} and the multilinear construction of Raz, Reingold, and Vadhan~\cite{RRV02} to obtain an explicit strong quantum-proof $(\mu n,\eta)$ extractor $F_{\mathrm{lin}}$ with seed length $d=O_\mu(h\log^2(n/h))$ and output length $m=\Omega_\mu(n)$ (Lemma~\ref{lem:long-output}).
We mention here that the extractor of \cite{RRV02} is an improved version of Trevisan's extractor and Trevisan's extractor has been used by \cite{DPVR} in their quantum-proof extractor construction. 

We first construct a strong quantum-proof extractor that produces a single block of $\Theta(h)$ bits. Choose $\lambda=\Theta(h)$, $\ell=\Theta(\log(n/h))$, and $m_0=\Theta(\lambda)$ with $n\leq\lambda2^\ell$ and $m_0\leq\lambda$. 
For a sample $x\sim X$, pad $x$ with zeroes and define a multilinear polynomial $p_x$ over $\F_{2^\lambda}$ in $\ell$ variables by using the bits of $x$ sequentially as the coefficients of $p_x$. Put $F(x,u)=p_x(u)$ and construct a two-universal hash family $\{h_w:\F_{2^\lambda}\to\{0,1\}^{m_0}\}_{w\in\F_{2^\lambda}}$.
The extractor we construct for this step is $h_w(F(x,u))$ where $x$ is the input source string, and the seed string is the concatenated string of $w$ and $u$.
%he extractor we construct here is  for $u\in\F_{2^\lambda}^{\ell}$, $w\in\F_{2^\lambda}$, the output $h_w(F(x,u))$.
Notice that the output length is $m_0 = O(h)$ and the seed length is $\lambda(\ell+1)=O(h\log(n/h))$.
We use the collision bound of~\cite[Lemma~5]{LHL} and the quantum leftover hash lemma~\cite[Theorem~5.5.1]{Ren05} to argue that this extractor has error at most
\begin{equation}\label{eq:overview-single-block}
 \frac12\sqrt{2^{m_0}\bigl(2^{-k}+\ell2^{-\lambda}\bigr)}
\end{equation}
for every source $\rho_{XE}$ with $\hmin(X\mid E)_\rho \geq k$.

To increase the output length to be linear of $k$, we use the standard weak design (from \cite{RRV02}) to generate $M$ blocks of bits from $v$ bits of the seed. 
Choose an $(\ell,\varrho)$ weak design $S_1,\ldots,S_M\subseteq[v]$, so that $|S_i|=\ell$ and $\sum_{j<i}2^{|S_i\cap S_j|}\leq\varrho M$ for every $i$. We take $\varrho=2$ and $\ell=\Theta(\log(n/h))$, and use an explicit design with $v=O(\ell^2)$ ~\cite[Lemma~15]{RRV02}. Choose $M=\Theta_\mu(n/\lambda)$ small enough that $\lambda\varrho M\leq\mu n/4$. With uniform random $Y\in\F_{2^\lambda}^v$ and $W\in \F_{2^\lambda}$, we construct the output blocks $Z_i=h_W(F(X,Y|_{S_i}))$. 
The output of $F_{\mathrm{lin}}$ is defined to be the concatenation of these $Z_i, i = 1, \ldots, M$.

Fix an output block $Z_i$. We show that $Z_i$ is close to uniform even conditioned on the preceding blocks $Z_1,\ldots,Z_{i-1}$, the seed $Y$, and the leaked quantum information $E$. Fix $Y|_{[v]\setminus S_i}$. For each $j<i$, the restricted polynomial $p_X(Y|_{S_j})$ depends only on the coordinates indexed by $S_i\cap S_j$. Reveal the coefficients of all restricted polynomials $p_X(Y|_{S_j})$  with $j<i$ as advice $A_i$ which is a classic random variable. The tuple $(A_i,Y,W)$ determines $Z_1,\ldots,Z_{i-1}$.

Consider each $p_X(Y|_{S_j})$ as a polynomial of the unfixed variables.
Then each $p_X(Y|_{S_j})$  has at most $2^{|S_i\cap S_j|}$ coefficients of $\lambda$ bits each. If $\hmin(X\mid E)_\rho\geq k$, the weak-design and classical leakage bounds give
\[
 |A_i|\leq\lambda\sum_{j<i}2^{|S_i\cap S_j|}\leq\lambda\varrho M,
 \qquad \hmin(X\mid A_iE)_\rho\geq k-\lambda\varrho M.
\]
The advice $A_i$ depends only on $X$ and the fixed coordinates outside $S_i$, i.e. $Y|_{[v]\setminus S_i}$. Thus $Y|_{S_i}$ and $W$ remain uniform and independent of $(X, A_i, E)$. We can apply the single-block extractor to source $X$ with seeds $Y|_{S_i},W$, while the side information is regarded as $(A_i, E)$. At $k=\mu n$, the remaining entropy is at least $3\mu n/4$. With our choices of $\lambda$ and $m_0$, \eqref{eq:overview-single-block} shows that $Z_i$ is $\eta/M$-close to uniform given $A_i,Y,W,E$. Since $A_i,Y,W$ determine the preceding output blocks, $Z_i$ is also $\eta/M$-close to uniform conditioned on $Z_1,\ldots,Z_{i-1},Y,W,E$.
By a hybrid argument, the total error can be at most $\eta$. Also notice that the output length is $M m_0=\Omega_\mu(n)$, and the seed length is the length of $Y$ plus the length of $W$ which is $\lambda(v+1)=O(h\log^2(n/h))$.

\paragraph{Reducing the seed length.}
The task in this stage is to further reduce the seed length. A similar situation also happens for the classic extractor construction, i.e. how to improve the seed length of \cite{RRV02} while somehow maintain the general construction structure of \cite{RRV02}. It happens that a similar situation is also encountered by \cite{CW24} for low-depth circuit construction of extractors. Indeed, we will follow a general strategy of \cite{CW24} to create a sequence of blocks having their lengths being a geometric sequence, and then use them to provide a long enough seed for the extractor in the previous stage.

We construct an extractor $G$ that uses $O_\mu(h)$ seed bits and produces $O_\mu(h\log^2(n/h))$ bits to seed the multilinear extractor of the previous stage. Fix a source $\rho_{XE}$ with $\hmin(X\mid E)_\rho\geq\mu n$. 
Specifically, using quantum entropy sampling (Lemma~\ref{lem:public-entropy-sampling}), we can use a common public seed $R$ to select blocks $B_0,\ldots,B_{T-1}$ of bits from the source $X$, where the lengths of $B_0, \ldots, B_{T-1}$ form an increasing geometric sequence. The selecting for each block is done by a standard source sampling procedure using a standard sampler. Also notice that we use the same $R$ for selecting each of these blocks. So the total seed length for such sampling is $O_\mu(h)$. The sampled blocks may overlap with each other. We let the smallest block to have length $\Theta_\mu(\sqrt{nh})$ to meet the entropy-sampling requirement.
%We managed to show that for every $i = 1, \ldots, T-1$, the block $B_i$ has a high enough min-entropy conditioned on previous blocks $B_1, \ldots, B_{i-1}$.
Each block $B_i$ has length $\lceil1+8/\mu\rceil^i\cdot\Theta_\mu(\sqrt{nh})$ for $0\leq i<T$.

Next we are going to use a short-seed extractor construction of  {\cite[Theorem~10]{LHL}}. 
%Lemma~\ref{lem:short-hash}. 
To obtain $m$ output bits with error $\eta/(2T)$, this kind of extractor requires $m+O(\log(T/\eta))$ bits of conditional min-entropy of the source, and at most $2m+O(h)$ seed bits. 
We start with an independent uniform $O_\mu(h)$-bit seed $Z_0$. We plan to repeatedly extend $Z_j$ to a longer string $Z_{j+1}$ close to uniform, for $j = 0, \ldots, T-1$, by processing the blocks $B_{T-1} , B_{T-2}, \ldots, B_1$. 
Let $H_j$ be an extractor constructed from {\cite[Theorem~10]{LHL}} with seed length $|Z_j|$ and output length $\lfloor|Z_j|/4\rfloor$. Apply $H_j$ to $B_{T-1-j}$ and append its output to the seed, i.e.:
\begin{equation}\label{eq:overview-seed-recurrence}
 Z_{j+1}=Z_j\circ H_j(B_{T-1-j},Z_j),
\end{equation}
where $\circ$ is for concatenation.
Each application of~\eqref{eq:overview-seed-recurrence} increases the seed length by a constant factor, so $T=O(\log\log(n/h))$ applications suffice. The external seed is $(R,Z_0)$. View $Z_T$ as $Z_0\circ W$. We define the output of extractor $G$ to be $W$, whose length is $O_\mu(h\log^2(n/h))$. This completes the construction of $G$.

The reverse order ensures that each block $B_i$ retains large enough entropy conditioned on $B_0, \ldots, B_{i-1}$. Write $B_{<i}=(B_0,\ldots,B_{i-1})$. We managed to show that the $B_i$ we attained from sampling has a constant smooth min-entropy rate conditioned on $ (R,E)$ (smooth min-entropy, which we denote by $\hmin^\tau$, is defined in~\eqref{eq:smooth-min-entropy}). Notice that $B_{<i}$ has at most $\sum_{q<i}|B_q|$ bits of entropy. By choosing the block lengths to be a geometric sequence, we obtain
\begin{equation}\label{eq:overview-block-entropy}
 \hmin^\tau(B_i\mid B_{<i}RE)_\rho \geq c_\mu|B_i|,
\end{equation}
where $c_\mu>0$ is a constant and $\tau=\eta/(8T)$. The smallest block has $\Omega_\mu(\sqrt{nh})$ bits of smooth conditional min-entropy, enough to inductively argue that the input of $H_j$ has enough input entropy, for $j = T-1, T-2, \ldots, 0$.
%for every extraction step when $n/h$ is sufficiently large.

We show by induction that, after $j$ applications of~\eqref{eq:overview-seed-recurrence}, $Z_j$ is $j\eta/T$-close to uniform even given $RE$ and the unused blocks $B_{<T-j}$. The claim holds for the initial seed $Z_0$. For $j<T$, put $i=T-1-j$. The induction hypothesis allows us to replace $Z_j$ by a uniform seed $V_j$ independent of $(B_i, B_{<i}, R, E)$, at a cost of at most $j\eta/T$ in trace distance. Applying the update in~\eqref{eq:overview-seed-recurrence} cannot increase the distance. By~\eqref{eq:overview-block-entropy}, the block $B_i$ has enough smooth min-entropy to apply $H_j$ with side information $(B_{<i},R,E)$. Lemma~\ref{lem:smooth-extraction} then shows that $V_j\circ H_j(B_i,V_j)$ is $\eta/T$-close to uniform given $(B_{<i}, R, E)$. The triangle inequality gives error at most $(j+1)\eta/T$ for $Z_{j+1}$, completing the induction.
Taking $j=T$, we obtain that $Z_T$ is $\eta$-close to uniform given $RE$. $Z_T=Z_0\circ W$, so $W$ is $\eta$-close to uniform even given $R,Z_0,E$.

Finally, we combine $G$ with the linear-output extractor $F_{\mathrm{lin}}$ attained in the previous paragraph. 
%We use a random string $A$ to select a set $S(A)\subseteq[n]$ containing a small constant fraction of the source coordinates. 
We use a standard source sampling process to attain  $B=X_{S(A)}$ and $C=X_{[n]\setminus S(A)}$, where $S\subseteq [n]$ is generated by a standard averaging sampler. With an additional seed $Y$, output
\[
 F_{\mathrm{lin}}\bigl(B,G(C,Y)\bigr).
\]
We managed to show that $B$ has a constant smooth min-entropy rate given $(A, E)$. Equation~\eqref{eq:classical-pair-chain} gives $C$ a constant conditional min-entropy rate even given $B, A, E$. Thus $G(C,Y)$ is nearly uniform given $B,A,Y,E$ and can serve as the seed for $F_{\mathrm{lin}}$. Applying $F_{\mathrm{lin}}$ to $B$ gives $\Omega_\mu(n)$ output bits with error at most $\eta$. This extractor only uses $O_\mu(h)$ seed bits since the seed consists of $Y$ and the randomness to generate $S$. This shows Theorem~\ref{thm:constant-rate}.

\paragraph{Extraction with small output loss.}
We combine our constant-rate extractor with a quantum-proof version of the Goldreich--Wigderson extractor. For sources of conditional min-entropy $n-a$, the combination extracts all but $O(a+\log(1/\eta))$ bits. 
Using the lifting theorem of Berta, Fawzi, and Scholz~\cite{BFS}, we managed to show that Goldreich--Wigderson extractor $H: \{0, 1\}^n \times \{0, 1\}^d \rightarrow \{0, 1\}^m$ is a quantum-proof extractor with seed length and output loss both being $O(a+\log(1/\eta))$.

Split the source into a block $B$ of $C_0(a+\log(1/\eta))$ bits and the remaining block $Z$, where $C_0$ is a sufficiently large constant. Apply our constant-rate extractor $F$ of Theorem~\ref{thm:constant-rate}, to $B$ with seed $Y$, and output
\[
 H\bigl(Z,F(B,Y)\bigr).
\]
Equation~\eqref{eq:classical-pair-chain} gives $\hmin(B\mid ZE)_\rho\geq|B|-a$ and $\hmin(Z\mid E)_\rho\geq|Z|-a$. Thus $B$ has a constant entropy rate even given $ZE$, and $F(B,Y)$ is nearly uniform given $ZYE$. The output $H(Z,F(B,Y))$ is therefore $\eta$-close to uniform given $YE$. The seed $Y$ has length $O(\log a+\log(1/\eta))$. The total output loss, including the reserved block $B$, remains $O(a+\log(1/\eta))$. This shows our Lemma~\ref{lem:deficiency-base}.

\paragraph{Further Refining the extractor's seed length.}
To shorten the extractor seed further (in order to achieve smaller sample complexity for the corresponding sampler), split the source into an $n$-bit block $X$ and an $s$-bit block $Z$. We apply the extractor $\Ext_0:\{0,1\}^n\times\{0,1\}^{d_0}\to\{0,1\}^m$ from Lemma~\ref{lem:deficiency-base} to $X$, with $d_0=O(\log a+\log(1/\eta))$ and error $\eta_0=\eta/2$. Following Xun and Zuckerman~\cite[Section~1.6.1]{XZ}, we use $Z$ to generate a list of candidate seeds for $\Ext_0$. The public seed selects one candidate.

Fix integers $\ell\geq1$ and $t\geq\ell$. Interpret $z$ as the coefficients of a degree-$(2\ell-1)$ polynomial $p_z$ over a sufficiently large binary field. Evaluate $p_z$ at distinct fixed points $u_1,\ldots,u_t$, and project each evaluation onto $d_0$ binary coordinates. Denote the resulting candidate seeds by $G(z,1),\ldots,G(z,t)$, where $G:\{0,1\}^s\times[t]\to\{0,1\}^{d_0}$. The coefficient block has length $s=O(\ell(d_0+\log t))$. For uniform coefficients, $G(U_s,1),\ldots,G(U_s,t)$ are $2\ell$-wise independent and uniform. With a public seed $J$ uniform on $[t]$ and independent of $XZE$, output
\[
 V=\Ext((X,Z);J)=\Ext_0(X;G(Z,J)).
\]

The candidate seed $G(Z,J)$ may be correlated with $X$ and $E$. We compare the actual output $\Ext_0(X;G(Z,J))$ with $\Ext_0(X;U)$, where $U$ is a uniform $d_0$-bit seed independent of $XZEJ$. Let $\rho_{VJE}$ and $\sigma_{VJE}$ denote the corresponding states, using $V$ for the $m$-bit output register in both. Assume $\hmin(XZ\mid E)\geq n+s-a$. Then $\hmin(X\mid E)\geq n-a$, so the guarantee for $\Ext_0$ gives
\begin{equation}\label{eq:overview-reference-output}
 \dist\bigl(\sigma_{VJE},\tau_m\otimes\rho_{JE}\bigr)\leq\eta_0.
\end{equation}
A $2\ell$th-moment argument bounds the effect of replacing $U$ by $G(Z,J)$ (see the proof of Theorem~\ref{thm:moment-tradeoff}):
\begin{equation}\label{eq:overview-seed-replacement}
 \dist\bigl(\rho_{VJE},\sigma_{VJE}\bigr)
 \leq 2^{a/(2\ell)}\sqrt{\frac{3\ell}{t}}.
\end{equation}

For $t=\Theta(\eta^{-2}\ell2^{a/\ell})$, equations~\eqref{eq:overview-reference-output} and~\eqref{eq:overview-seed-replacement} give $\dist(\rho_{VJE},\tau_m\otimes\rho_{JE})\leq\eta$. The public seed length is
\[
 d=\log t\leq2\log(1/\eta)+\log\ell+\frac a\ell+O(1).
\]
Increasing $\ell$ reduces the term $a/\ell$ but enlarges the coefficient block, whose $s$ bits contribute to the output loss. Balancing these costs gives $d\leq(1+\alpha)\log a+(2+\alpha)\log(1/\eta)+O_\alpha(1)$ while keeping the total output loss $O_\alpha(a+\log(1/\eta))$. This shows Theorem~\ref{thm:deficiency}.

\paragraph{Sampler parameters from the extractor.}
Apply the correspondence~\cref{thm:correspondence} with deficiency $a=\lceil\log(4D/\delta)\rceil$ and extraction error $\eta=\varepsilon/4$. The extractor's output loss becomes the sampler's randomness overhead, and its $d$-bit seed gives $2^d$ samples. Choosing $\ell$ to balance the seed length against the output loss yields the sample complexity in Theorem~\ref{thm:extractor-sampling}, with randomness $m+O_\alpha(\log(D/\delta))$.

%% file: preli.tex
\section{Preliminaries}\label{sec:prelim}
All logarithms are base two. Logarithmic factors in asymptotic bounds are understood to be at least one. For entropy deficiency $a\geq0$, we write $\log a=\log_2\max\{2,a\}$. Sampler parameter comparisons use $0<\delta\leq\varepsilon\leq1/2$. For a finite set $\mathcal A$, $U_{\mathcal A}$ and $\tau_{\mathcal A}$ denote the uniform distribution and density operator on $\mathcal A$. For classical distributions, $\dist$ denotes total variation distance.

We use $n,d,m$ for extractor source, seed, and output lengths. The output loss is $n-m$, and the entropy deficiency is $n-k$. 
%Both extractor models in Definition~\ref{def:extractor} allow a uniform seed on $[t]$. 
%The uniform state $\tau_{[t]}$ then replaces $\tau_d$ in~\eqref{eq:strong-extraction}. 
Sampler randomness complexity and sample complexity are $r$ and $t$, respectively. The matrix dimension is $D$. Sampler failure probability and accuracy are $\delta$ and $\varepsilon$. Extraction error is $\eta$. The bounded-storage parameter is $b$, so $\dim E\leq2^b$.

\subsection{Hermitian dilation}
The Hermitian dilation of $A\in\mathbb C^{D\times D}$ is
\[
 \mathcal H(A)=\begin{pmatrix}0&A\\A^*&0\end{pmatrix}.
\]
The identity
\[
 \mathcal H(A)^2=\begin{pmatrix}AA^*&0\\0&A^*A\end{pmatrix}
\]
gives $\norm{\mathcal H(A)}=\norm A$.
%Conjugation by $\operatorname{diag}(I,-I)$ changes $\mathcal H(A)$ to $-\mathcal H(A)$. 
Notice that the conjugation of $\mathcal H(A)$ by $\operatorname{diag}(I,-I)$, i.e. $\operatorname{diag}^{-1}(I,-I) \mathcal H(A) \operatorname{diag}(I,-I)$, is equal to $-\mathcal H(A)$.
So the eigenvalues of $\mathcal H(A)$ occur in opposite pairs. The largest eigenvalue of $\mathcal H(A)$ is $\norm A$. The map $\mathcal H$ is linear.

\subsection{Quantum side information and extraction}
A quantum state on a finite-dimensional register is represented by a positive semidefinite operator $\rho$. It is normalized if $\Tr\rho=1$ and subnormalized if $\Tr\rho\leq1$. States are normalized unless explicitly described as subnormalized. The trace norm is $\norm A_1=\Tr\sqrt{A^*A}$.

A classical--quantum (cq) state $\rho_{XE}$ has a classical register $X$ (which is a random variable) and a quantum register $E$. We write it using subnormalized blocks $\rho^x_E, x\in \supp(X)$:
\[
 \rho_{XE}=\sum_x\ket{x}\bra{x}\otimes\rho_E^x,\qquad
 \rho_E=\sum_x\rho_E^x.
\]

The min-entropy of a classical random variable $X$ is
$H_\infty(X):=-\log\max_x\Pr[X=x]$.
For a cq state $\rho_{XE}$, we define
$H_\infty(X)_\rho:=-\log\max_x\Tr\rho_E^x$.

A positive operator-valued measurement (POVM) is a family of positive semidefinite matrices $(M_x)_{x\in \supp(X)}$ with $\sum_xM_x=I$. Using the outcome of $(M_x)_x$ as a guess for $X$ gives success probability $\sum_x\Tr(M_x\rho_E^x)$. 
We define the optimal guessing probability (guessing $X$ correctly based on a measurement on $E$) to be
\begin{equation}\label{eq:pgood}
    \pg(X\mid E)_\rho:=\max_{(M_x)_x}\sum_x\Tr(M_x\rho_E^x),
\end{equation}
where the maximum is over POVMs on $E$.
The conditional min-entropy is defined by
$\hmin(X\mid E)_\rho=-\log\pg(X\mid E)_\rho$.
Equivalently, conditional min-entropy has the semidefinite characterization~\cite[Theorem~1]{KRS}
\begin{equation}\label{eq:entropy-domination}
 2^{-\hmin(X\mid E)}=\min\{\Tr W_E: W_E\succeq\rho_E^x\text{ for every }x\}.
\end{equation}
Thus $\hmin(X\mid E)\geq k$ gives a normalized state $\omega_E$ with $\rho_E^x\preceq2^{-k}\omega_E$ for all $x$. 
%The reference state $\omega_E$ need not be the actual marginal $\rho_E$. 
Equation~\eqref{eq:entropy-domination} also defines the conditional min-entropy for subnormalized cq states $(X, E)$.

A matrix $M$ with $0\preceq M\preceq I$ defines a binary measurement $\{M,I-M\}$. We accept on outcome $M$, with probability $\Tr(M\rho)$ on state $\rho$. For normalized states, we define the trace distance of two normalized states to be
\begin{equation}\label{eq:trace-distance-test}
 \dist(\rho,\sigma):=\max_{0\preceq M\preceq I}|\Tr M(\rho-\sigma)|.
\end{equation}

A quantum channel is a completely positive, trace-preserving map. Trace distance is contractive under quantum channels: $\dist(\Phi(\rho),\Phi(\sigma))\leq\dist(\rho,\sigma)$ for every quantum channel $\Phi$.

For a state $\rho_{XVE}$ with classical registers $X,V$ and a quantum register $E$, where $V$ contains at most $a$ bits, we have the leakage inequality
\begin{equation}\label{eq:chain}
 \hmin(X\mid VE)_\rho\geq\hmin(X\mid E)_\rho-a,
 \qquad \hmin(X\mid E)_\rho\geq H_\infty(X)_\rho-\log\dim E,
\end{equation}
where 2nd inequality of \eqref{eq:chain} follows from~\eqref{eq:entropy-domination} by taking $(\max_x\Pr[X=x])I_E$ as a feasible operator.
For a state $\rho_{X_1X_2E}$ a pair of classical registers $(X_1, X_2)$ and a quantum register $E$, we also have the chain inequality
\begin{equation}\label{eq:classical-pair-chain}
 \hmin(X_1\mid E)_\rho\geq\hmin(X_1X_2\mid E)_\rho-|X_2|,
 \qquad \hmin(X_2\mid X_1E)_\rho\geq\hmin(X_1X_2\mid E)_\rho-|X_1|.
\end{equation}
The above leakage inequality and chain inequalities follow from Renner's min-entropy inequalities~\cite[Lemmas~3.1.9--3.1.10]{Ren05}.

\paragraph{Smooth min-entropy.}
We will define the smooth min-entropy of $\rho_{XE}$. 
%is the largest conditional min-entropy attainable by perturbing the joint state within a prescribed distance. We allow the obtained nearby state to be subnormalized.
%We measure the allowed perturbation using purified distance. 
For subnormalized states $\rho,\sigma$, the generalized fidelity and the purified distance are defined by
\begin{align*}
 \overline F(\rho,\sigma)
 &=\|\sqrt\rho\sqrt\sigma\|_1
   +\sqrt{(1-\Tr\rho)(1-\Tr\sigma)},\\
 P(\rho,\sigma)&=\sqrt{1-\overline F(\rho,\sigma)^2}.
\end{align*}

Purified distance is contractive under quantum channels~\cite[Lemma~7]{TCR}, i.e. for subnormalized states $\rho,\sigma$ and a quantum channel $\Phi$, one has $P(\Phi(\rho),\Phi(\sigma))\leq P(\rho,\sigma)$. 
%Purified distance also has an extension property~\cite[Corollary~9]{TCR} stated as the following equality \eqref{eq:purified-extension}. Let $\rho_{AB}$ and $\sigma_A$ be subnormalized states. There is a subnormalized extension $\sigma_{AB}$ of $\sigma_A$ satisfying
%\begin{equation}\label{eq:purified-extension}
% \Tr_B\sigma_{AB}=\sigma_A,\qquad
% P(\sigma_{AB},\rho_{AB})=P(\sigma_A,\rho_A),
% \qquad \rho_A=\Tr_B\rho_{AB}.
%\end{equation}

For subnormalized states, the trace-norm and purified-distance bounds of~\cite[Lemma~6]{TCR} imply
\begin{equation}\label{eq:purified-trace}
 \frac12\|\tilde\rho-\rho\|_1
 \leq P(\tilde\rho,\rho)
 \leq\sqrt{2\|\tilde\rho-\rho\|_1}.
\end{equation}

For a cq state $\rho_{XE}$, define the smooth min-entropy to be
\begin{equation}\label{eq:smooth-min-entropy}
 \hmin^\tau(X\mid E)_\rho
 =\sup_{\tilde\rho_{XE}:\,P(\tilde\rho_{XE},\rho_{XE})\leq\tau}
      \hmin(X\mid E)_{\tilde\rho},
\end{equation}
where the supremum ranges over subnormalized cq states on the same registers.

For trace-norm smoothing, we write
\begin{equation}\label{eq:trace-smooth-min-entropy}
 \hmin^{s,\|\cdot\|_1}(X\mid E)_\rho
 =\sup_{\tilde\rho_{XE}:\,\|\tilde\rho_{XE}-\rho_{XE}\|_1\leq s}
      \hmin(X\mid E)_{\tilde\rho},
\end{equation}
with the same convention on $\tilde\rho_{XE}$. Equation~\eqref{eq:purified-trace} gives
\begin{equation}\label{eq:smoothing-conversion}
 \hmin^{\sqrt{2s}}(X\mid E)_\rho
 \geq\hmin^{s,\|\cdot\|_1}(X\mid E)_\rho.
\end{equation}

The next two lemmas are proved in Appendix~\ref{app:smooth-proofs}.

\begin{lemma}[Smooth classical leakage; cf.~{\cite[Lemma~B.3]{DPVR}}]\label{lem:smooth-classical-leakage}
Let $\rho_{XLE}$ be a normalized state with classical registers $X,L$. Suppose $L$ takes values in a finite alphabet $\mathcal L$. For every $0\leq\tau<1$,
\[
 \hmin^\tau(X\mid LE)_\rho
 \geq\hmin^\tau(X\mid E)_\rho-\log|\mathcal L|.
\]
\end{lemma}

\begin{lemma}[Extraction from smooth sources; cf.~{\cite[Lemma~3.5]{DPVR}}]\label{lem:smooth-extraction}
Let $\Ext:\{0,1\}^n\times\{0,1\}^d\to\{0,1\}^m$ be a strong quantum-proof $(k,\eta)$ extractor. Let $\tau\geq0$, and let $\rho_{XE}$ be a normalized cq state satisfying $\hmin^\tau(X\mid E)_\rho\geq k+1$. For every independent uniform seed $Y\sim U_d$,
\[
 \dist\!\left(\rho_{\Ext(X,Y)YE},
       \tau_m\otimes\tau_d\otimes\rho_E\right)
 \leq\eta+4\tau.
\]
\end{lemma}

\paragraph{Two-universal hashing.}
A family $\{h_w:\mathcal Z\to\{0,1\}^m\}_{w\in\mathcal W}$ is \emph{two-universal} if, for a uniform seed $W$ on $\mathcal W$,
\[
 \Pr_W[h_W(z)=h_W(z')]\leq2^{-m}
 \qquad\text{for every distinct }z,z'\in\mathcal Z.
\]

%% file: extractor_sampler_correspondence.tex
\section{The Quantum Extractor--Matrix Sampler Correspondence}\label{sec:bridge}

We now prove Theorem~\ref{thm:correspondence}, restated below.

\begin{correspondencerestatement}[restated]
\input{statements/correspondence}
\end{correspondencerestatement}

\begin{proof}[Proof of Theorem~\ref{thm:correspondence}]
For part~(1), for every $i\in [t]$, fix tests $f_i:\{0,1\}^m\to\mathbb C^{D\times D}$ with $\norm{f_i(z)}\leq1$. For $x\in\{0,1\}^n$, write
\begin{equation}\label{eq:bridge-deviation}
 A_x=\frac1t\sum_{i=1}^t\bigl(f_i(F(x,i))-\E_{Z\sim U_m}f_i(Z)\bigr).
\end{equation}
Suppose $\mathcal B=\{x:\norm{A_x}>\varepsilon\}$ has size greater than $\delta2^n$. For $x\in\mathcal B$, the Hermitian dilation $\mathcal H(A_x)$ has a positive eigenvalue $\norm{A_x}>\varepsilon$. Choose a corresponding unit eigenvector $v_x$. Let $X$ be uniform on $\mathcal B$. Define $\rho_{XE}=|\mathcal B|^{-1}\sum_{x\in\mathcal B}\ket{x}\bra{x}\otimes\ket{v_x}\bra{v_x}$, where $E$ has dimension $2D$. The encoding $x\mapsto\ket{v_x}\bra{v_x}$ is independent of the index $i$. The source $X$ has min-entropy
\[
 H_\infty(X)=\log|\mathcal B|>n-\log(1/\delta).
\]
Let $I$ be uniform on $[t]$ and independent of $XE$. For every $z\in \{0, 1\}^m$, the matrix $\mathcal H(f_i(z))$ is Hermitian and has norm at most one. Define the operators
\begin{equation}\label{eq:bridge-measurement}
 M_{i,z}:=(I_{2D}+\mathcal H(f_i(z)))/2, \quad \forall i\in [t], z \in \{0, 1\}^m.
\end{equation}
This
satisfies $0\preceq M_{i,z}\preceq I_{2D}$ for every $i, z$.

Define the measurement operator
\[
 M:=\sum_{z\in\{0,1\}^m}\sum_{i=1}^t
 |z\rangle\langle z|\otimes|i\rangle\langle i|\otimes M_{i,z}.
\]
Since $0\preceq M_{i,z}\preceq I_E$ for every $i,z$,
$\{M,I-M\}$ is a POVM. 
Let the outcome corresponding to $M$ be acceptance of the POVM.
Then the acceptance probabilities of doing the POVM on
$\rho_{F(X,I)IE}$ and $\tau_m\otimes\rho_{IE}$ differ by
\begin{align*}
 &\Tr\!\left[M\left(\rho_{F(X,I)IE}
                  -\tau_m\otimes\rho_{IE}\right)\right]\\
 &\quad=\E_X\frac1t\sum_{i=1}^t
 \left\langle v_X,
 \left(M_{i,F(X,i)}-\E_{Z\sim U_m}M_{i,Z}\right)v_X
 \right\rangle\\
 &\quad=\frac12\E_X\langle v_X,\mathcal H(A_X)v_X\rangle
 =\frac12\E_X\norm{A_X}
 >\varepsilon/2.
\end{align*}
The second equality follows from~\eqref{eq:bridge-measurement}
and~\eqref{eq:bridge-deviation}, the third from the choice of
$v_x$, and the final inequality from the fact that $X$ is
supported on $\mathcal B$.
By~\eqref{eq:trace-distance-test},
\[
 \Delta({\rho_{F(X,I)IE},\tau_m\otimes\rho_{IE}})
 \geq
 \Tr\!\left[M\left(\rho_{F(X,I)IE}
                   -\tau_m\otimes\rho_{IE}\right)\right]
 >\varepsilon/2,
\]
contradicting the extractor guarantee
$\Delta(\rho_{F(X,I)IE},\tau_m\otimes\rho_{IE})
\leq\varepsilon/2$.
This shows the strong extractor and strong sampler case of part (1).

For the weak extractor and ordinary sampler case of part (1),
notice that if $f_1=\cdots=f_t$, then $M_{i,z}=M_z$ is independent of $i$.
The same acceptance-probability gap is therefore achieved by
the doing POVM $\{M',I-M'\}$ on $\rho_{F(X,I)E}$ and 
$\tau_m\otimes\rho_E$ , where
\[
 M':=\sum_{z\in\{0,1\}^m}|z\rangle\langle z|\otimes M_z.
\]
This measurement distinguishes $\rho_{F(X,I)E}$ from
$\tau_m\otimes\rho_E$ with an acceptance-probability gap
greater than $\varepsilon/2$, contradicting the weak
extractor guarantee
$\Delta(\rho_{F(X,I)E},\tau_m\otimes\rho_E)\leq\varepsilon/2$.

%Now for every $i$ and $z$, consider a procedure which measures $E$ with $\{M_{i,z},I_{2D}-M_{i,z}\}$ and accept if the outcome corresponds to $M_{i,z}$.

%Conditioned on $X=x$ and $I=i$, by our definition of $\rho_{XE}$, the register $E$ is in the state $\ket{v_x}\bra{v_x}$. The probability of measuring $E$ with $\{ M_{i,F(x,i)} , I-M_{i,F(x,i)}\}$ on the conditioned state $\rho_{F(X,I)IE|X=x,I=i}$ and attaining the outcome corresponding to $M_{i,F(x,i)}$, is therefore $\langle v_x,M_{i,F(x,i)}v_x\rangle$. On the other hand, by replacing $F(X,I)$ with an independent uniform $Z\sim U_m$, we obtain  ${ \rho_{ZIE}}_{|I=i,X=x}$. The probability of measuring $E$ with $\{ M_{i,Z} , I-M_{i, Z}\}$ on ${ \rho_{ZIE}}_{|I=i,X=x}$ and attaining the outcome corresponding to $M_{i, Z}$,  is $\E_{Z\sim U_m}\langle v_x,M_{i,Z}v_x\rangle$. For each $x\in\mathcal B$, averaging their difference over $i$ gives
%\begin{align*}
% \frac1t\sum_{i=1}^t\left\langle v_x,\left(M_{i,F(x,i)}-%\E_{Z\sim U_m}M_{i,Z}\right)v_x\right\rangle
% &=\frac12\langle v_x,\mathcal H(A_x)v_x\rangle\\
% &=\frac12\norm{A_x}>\varepsilon/2.
%\end{align*}
%The first equality follows from~\eqref{eq:bridge-measurement} and~\eqref{eq:bridge-deviation}; the second uses the choice of $v_x$. Averaging over $X$ gives the difference in acceptance probabilities on $\rho_{F(X,I)IE}$ and $\tau_m\otimes \rho_{IE}$, which is still greater than $\varepsilon/2$. By~\eqref{eq:trace-distance-test}, this contradicts the bounded-storage extraction guarantee. If $f_1=\cdots=f_t$, then $M_{i,z}$ is independent of $i$. Weak extraction then suffices.

For part~(2), consider a cq state with $H_\infty(X)\geq n-a$ and $\dim E\leq D$. Pad $E$ to dimension $D$. Write the state as $\rho_{XE}=\sum_xp_x\ket{x}\bra{x}\otimes\rho_x$, where  $p_x $ here denotes $ \Pr[X = x]\leq2^{a-n}$.  Consider a POVM $\{M,I_{2^m}\otimes I_t\otimes I_D-M\}$ for measuring $\rho_{F(X,I)IE}$ or
$\tau_m\otimes\rho_{IE}$
which has the form
\[
 M=\sum_{z\in\{0,1\}^m}\sum_{i=1}^t
 |z\rangle\langle z|\otimes|i\rangle\langle i|\otimes M_{i,z},
 \qquad 0\preceq M_{i,z}\preceq I_E.
\]
We define acceptance of the measure to be the case when the outcome corresponds to $M$.

Set $f_i(z)=2M_{i,z}-I$. The matrices $f_i(z)$ are Hermitian and satisfy $\norm{f_i(z)}\leq1$. Apply the strong sampler to $f_1,\ldots,f_t$. For each $x\in\{0,1\}^n$, define the operator
\[
 B_x:=\frac1t\sum_{i=1}^t\bigl(M_{i,F(x,i)}-\E_{Z\sim U_m}M_{i,Z}\bigr).
\]
Define $\mathcal B=\{x\in\{0,1\}^n:\norm{B_x}>\varepsilon/2\}$. The strong sampler guarantee gives $|\mathcal B|\leq\delta2^n$. %The factor $1/2$ arises because the deviation for the tests $f_i$ equals $2B_x$.
Both the two averages in the definition of $B_x$ are in between $0$ and $I$. Thus $-I\preceq B_x\preceq I$ for every $x$. Moreover,
\[
 \Pr[X\in\mathcal B]=\sum_{x\in\mathcal B}p_x\leq2^{a-n}|\mathcal B|\leq2^a\delta.
\]

The difference of acceptance probabilities for measuring $\rho_{F(X,I)IE}$ and $\tau_m\otimes\rho_{IE}$, is
\begin{align*}
 &\Tr\!\left[M\left(\rho_{F(X,I)IE}
                  -\tau_m\otimes\rho_{IE}\right)\right]\\
 &\quad=\sum_xp_x\frac1t\sum_{i=1}^t
 \left(
 \Tr(\rho_xM_{i,F(x,i)})
 -\E_{Z\sim U_m}\Tr(\rho_xM_{i,Z})
 \right)\\
 &\quad=\sum_xp_x\Tr(\rho_xB_x),
\end{align*}
where the first equality follows from the block-diagonal
form of $M$, and the second from the definition of $B_x$.

The absolute difference in acceptance probabilities is therefore
\begin{align*}
 \left|\sum_xp_x\Tr(\rho_xB_x)\right|
 &\leq\sum_{x\notin\mathcal B}p_x\norm{B_x}+\sum_{x\in\mathcal B}p_x\norm{B_x}\\
 &\leq\frac\varepsilon2\Pr[X\notin\mathcal B]+\Pr[X\in\mathcal B]
 \leq\frac\varepsilon2+2^a\delta.
\end{align*}
Since this holds for all measurements $M$, it proves the strong extractor and strong sampler case of part~(2). 

For weak extractors, consider a binary POVM $\{M,I_{2^m}\otimes I_D- M\}$with 
\[
 M=\sum_{z\in\{0,1\}^m}|z\rangle\langle z|\otimes M_z.
\]
Set $f(z)=2M_z-I_E$. Applying the ordinary sampler guarantee
to $f$ gives the same bound on the acceptance-probability
difference between $\rho_{F(X,I)E}$ and
$\tau_m\otimes\rho_E$, proving part~(2) for weak extraction.
\end{proof}

\begin{corollary}\label{cor:conditional-forward}
Let $F:\{0,1\}^n\times\{0,1\}^d\to\{0,1\}^m$ be a
strong (resp.\ weak) quantum-proof $(k,\eta)$ extractor.
Then $F$ is a strong (resp.\ ordinary) $D$-dimensional
$(\delta,\varepsilon)$ matrix sampler whenever
\begin{equation}\label{eq:conditional-forward}
 k\leq n-\log(2D/\delta),\qquad \eta\leq\varepsilon/2.
\end{equation}
The sampler uses $r=n$ random bits and has sample complexity
$t=2^d$.
\end{corollary}
\begin{proof}
Let $\rho_{XE}$ be a cq state with
$H_\infty(X)_\rho\geq n-\log(1/\delta)$ and $\dim E\leq2D$.
By the second inequality in~\eqref{eq:chain},
\[
 \hmin(X\mid E)_\rho
 \geq H_\infty(X)-\log\dim E
 \geq n-\log(2D/\delta)
 \geq k.
\]
Thus $F$ is a strong (resp.\ weak) $(\log 2D)$-qubit-proof 
$(n-\log(1/\delta),\varepsilon/2)$ extractor.
Applying Theorem~\ref{thm:correspondence}, part~(1),
gives the claimed sampler.
\end{proof}

%% file: quantum_extractor.tex
\section{A Quantum-Proof Extractor with Linear Output}\label{sec:design}
Fix a constant conditional min-entropy rate $\mu>0$. Put $h=\log(n/\eta)$ and assume that $h$ is at most a sufficiently small constant multiple of $n$. We construct a linear-output quantum-proof extractor with seed length $O_\mu(h\log^2(n/h))$.

\subsection{Collision bounds and hashing}\label{sec:collision}
We combine the collision bound of \cite[Lemma~5]{LHL} with the quantum leftover hash lemma~\cite[Theorem~5.5.1]{Ren05}.

\begin{lemma}\label{thm:collision}
Let $F:\mathcal X\times\mathcal Y\to\mathcal Z$ satisfy
\[
 \Pr_{Y\sim U_{\mathcal Y}}[F(x,Y)=F(x',Y)]\leq\xi
 \qquad\text{for all distinct }x,x'\in\mathcal X.
\]
Let $k$ be an entropy threshold and let $\eta>0$. Let $\{h_w:\mathcal Z\to\{0,1\}^m\}_{w\in\mathcal W}$ be a two-universal hash family, and suppose
\begin{equation}\label{eq:collision-extraction}
 \frac12\sqrt{2^m(2^{-k}+\xi)}\leq\eta.
\end{equation}
Then $\Ext:\mathcal X\times(\mathcal Y\times\mathcal W)\to\{0,1\}^m$, defined by $\Ext(x;y,w)=h_w(F(x,y))$, is a strong quantum-proof $(k,\eta)$ extractor with a uniform seed on $\mathcal Y\times\mathcal W$.
\end{lemma}
\begin{proof}
For a normalized cq state $\omega_{AE}=\sum_a\ket a\bra a\otimes\omega_E^a$, define
\[
 \cp(A\mid E)_\omega
 :=\sum_a\Tr\!\left(\omega_E^a\omega_E^{-1/2}\omega_E^a\omega_E^{-1/2}\right),
 \qquad \omega_E=\sum_a\omega_E^a,
\]
where, if $\omega_E=\sum_j\lambda_j\ket{e_j}\bra{e_j}$ is a
spectral decomposition, we define
\[
 \omega_E^{-1/2}
 :=\sum_{j:\lambda_j>0}\lambda_j^{-1/2}\ket{e_j}\bra{e_j}.
\] 
The operators
\[
 \left\{M_a:=\omega_E^{-1/2}\omega_E^a\omega_E^{-1/2}
 \right\}_{a\in\mathcal A}
\]
form a POVM, where $\mathcal A$ is the
alphabet of $A$.
The probability that this measurement outputs the correct value of $A$ is $\sum_a\Tr(M_a\omega_E^a)=\cp(A\mid E)_\omega$. Hence $\cp(A\mid E)_\omega\leq\pg(A\mid E)_\omega$. Recall that $\pg$ is defined by equation \eqref{eq:pgood}.

Fix a source $\rho_{XE}$ with $\hmin(X\mid E)\geq k$, and put $Z=F(X,Y)$ for an independent uniform seed $Y$. Applying~\cite[Lemma~5, Eq.~(17)]{LHL} gives
\[
 \cp(Z\mid YE)_\rho
 \leq\cp(X\mid E)_\rho+\xi
 \leq2^{-k}+\xi.
\]
The quantum leftover hash lemma in collision form~\cite[Theorem~5.5.1]{Ren05}, with side information $YE$ and an independent uniform hash seed $W$, now yields
\[
 \dist\!\left(\rho_{h_W(Z)WYE},
       \tau_m\otimes\rho_W\otimes\rho_{YE}\right)
 \leq\frac12\sqrt{2^m\cp(Z\mid YE)_\rho}
 \leq\frac12\sqrt{2^m(2^{-k}+\xi)}.
\]
\end{proof}

\subsection{Multilinear evaluations and weak designs}
We adapt the extractor construction of Raz, Reingold, and Vadhan~\cite{RRV02} to sources with quantum side information. %Encode the source string as a multilinear polynomial. To produce each output block, evaluate the polynomial at a point and hash the result. We save randomness by taking all evaluation points from coordinate restrictions of one random vector. A weak design controls the overlaps among the coordinate sets.

\begin{definition}[Weak design]
A family $S_1,\ldots,S_M\subseteq[v]$ is an $(\ell,\varrho)$ weak design if $|S_i|=\ell$ for all $i$ and $\sum_{j<i}2^{|S_i\cap S_j|}\leq\varrho M$ for every $i$.
\end{definition}
\begin{theorem}[{\cite[Lemma~15]{RRV02}}]\label{thm:design-size}
For every $\varrho>1$, there is an explicit design with
\begin{equation}\label{eq:design-size}
 v=O\!\left(\ell+\frac{\ell^2}{\log\varrho}\right).
\end{equation}
The design can be constructed in time polynomial in $M$ and $v$.
\end{theorem}

\begin{theorem}\label{thm:design}
Let $n,\lambda,\ell,M,m_0\geq1$ be integers with $n\leq\lambda2^\ell$ and $m_0\leq\lambda$. Let $S_1,\ldots,S_M\subseteq[v]$ be an $(\ell,\varrho)$ weak design. For every $0\leq k\leq n$ and $\eta>0$ satisfying
\begin{equation}\label{eq:design-error}
 \frac M2\sqrt{2^{m_0}\left(2^{\lambda\varrho M-k}+\ell2^{-\lambda}\right)}\leq\eta,
\end{equation}
there is a strong quantum-proof $(k,\eta)$ extractor $\Ext:\{0,1\}^n\times\{0,1\}^d\to\{0,1\}^m$ with $d=\lambda(v+1)$ and $m=M m_0$.
\end{theorem}
\begin{proof}
Fix a binary basis of $\F_{2^\lambda}$. On input $x\in\{0,1\}^n$, the extractor proceeds as follows.
\begin{enumerate}
\item \textbf{Polynomial encoding.} Pad $x$ with zeroes to $\lambda2^\ell$ bits. Divide the padded string into $2^\ell$ blocks of $\lambda$ bits. Index the blocks by subsets $I\subseteq[\ell]$ in a fixed order. Interpret each block as a field element $a_I(x)$. Define
\[
 p_x(z_1,\ldots,z_\ell)=\sum_{I\subseteq[\ell]}a_I(x)\prod_{j\in I}z_j \in\mathbb F_{2^\lambda}[z_1,\ldots,z_\ell].
\]
Distinct strings $x$ give distinct formal polynomials $p_x$.
\item \textbf{Generating the seed.} Let $(Y,W)$ be uniform on $\F_{2^\lambda}^v\times\F_{2^\lambda}$. For each $i$, the point $Y|_{S_i}\in\F_{2^\lambda}^\ell$ lists the coordinates indexed by $S_i$ in increasing order. The same pair $(Y,W)$ is used for all output blocks.
\item \textbf{Evaluation and hashing.} Fix a linear projection $\pi_{m_0}:\F_{2^\lambda}\to\{0,1\}^{m_0}$; for example, retain the first $m_0$ basis coordinates. For each $i$, evaluate $p_x$ at $Y|_{S_i}$, multiply the result by $W$, and apply $\pi_{m_0}$. Concatenating the blocks gives
\[
 \Ext(x;Y,W)=\bigl(\pi_{m_0}(Wp_x(Y|_{S_1})),\ldots,\pi_{m_0}(Wp_x(Y|_{S_M}))\bigr).
\]
\end{enumerate}

Fix a normalized cq source $\rho_{XE}$ with $\hmin(X\mid E)_\rho\geq k$. Extend $\rho_{XE}$ by the independent uniform seed $(Y,W)$ and the output $Z=(Z_1,\ldots,Z_M)$. Denote the resulting joint state by $\rho$. The output blocks are
\[
 Z_i=\pi_{m_0}(Wp_X(Y|_{S_i})).
\]
Write $Z_{\leq i}=(Z_1,\ldots,Z_i)$ and $Z_{<i}=(Z_1,\ldots,Z_{i-1})$, with $Z_{\leq0}$ and $Z_{<1}$ empty. Set
\[
 \varepsilon_0=\frac12\sqrt{2^{m_0}\left(2^{\lambda\varrho M-k}+\ell2^{-\lambda}\right)}.
\]

\emph{Hybrid states.} For $0\leq i\leq M$, replace $Z_{i+1},\ldots,Z_M$ by independent uniform blocks to obtain
\[
 \rho^{(i)}_{ZYWE}
 =\rho_{Z_{\leq i}YWE}\otimes\tau_{m_0}^{\otimes(M-i)}.
\]
We abbreviate $\rho^{(i)}_{ZYWE}$ by $\rho^{(i)}$. The endpoints are
\[
 \rho^{(0)}_{ZYWE}=\tau_m\otimes\rho_{YWE},\qquad
 \rho^{(M)}_{ZYWE}=\rho_{ZYWE}.
\]
For $1\leq i\leq M$, the states $\rho^{(i)}$ and $\rho^{(i-1)}$ have the same independent uniform suffix $(Z_{i+1},\ldots,Z_M)$. Discarding $(Z_{i+1},\ldots,Z_M)$ preserves the trace distance between $\rho^{(i)}$ and $\rho^{(i-1)}$. Thus
\[
 \dist(\rho^{(i)},\rho^{(i-1)})
 =\dist\!\left(\rho_{Z_iZ_{<i}YWE},
       \tau_{m_0}\otimes\rho_{Z_{<i}YWE}\right).
\]
We will prove $\dist(\rho^{(i)},\rho^{(i-1)})\leq\varepsilon_0$ for every $i\in[M]$.

\emph{Distance between consecutive hybrids.} Fix $i$, and split $Y$ into $Y_{\mathrm{in}}=Y|_{S_i}$ and $Y_{\mathrm{out}}=Y|_{[v]\setminus S_i}$. For fixed $y_{\mathrm{out}}$, let
\[
 \sigma_{XEYW}:
 =\rho_{XEYW\mid Y_{\mathrm{out}}=y_{\mathrm{out}}}
\]
Conditioning on $Y_{\mathrm{out}}=y_{\mathrm{out}}$ leaves the marginal on $XE$ unchanged, so $\sigma_{XE}=\rho_{XE}$. The seeds $Y_{\mathrm{in}},W$ are independent and uniform of $\sigma_{XE}$

For each $j<i$, fixing $Y_{\mathrm{out}}$ fixes all coordinates
of $Y|_{S_j}$ outside $S_i\cap S_j$. Thus $p_X(Y|_{S_j})$ becomes a multilinear polynomial $q_j$
in the remaining coordinates $Y|_{S_i\cap S_j}$.
Let $A_i$ contain the coefficients of $q_j$ for all $j<i$. There are $2^{|S_i\cap S_j|}$ field coefficients for each $j$. Each field coefficient is represented by $\lambda$ bits. Thus $A_i$ has at most $\lambda\sum_{j<i}2^{|S_i\cap S_j|}\leq\lambda\varrho M$ bits. The classical leakage inequality~\eqref{eq:chain} gives
\[
 \hmin(X\mid A_iE)_\sigma\geq k-\lambda\varrho M.
\]
The advice $A_i$ depends only on $X$ and the fixed $y_{\mathrm{out}}$. Hence $Y_{\mathrm{in}}$ and $W$ remain uniform and independent of $XA_iE$ in $\sigma$.

Define
\[
 F:\{0,1\}^n\times\F_{2^\lambda}^{\ell}\to\F_{2^\lambda},
 \qquad F(x,u)=p_x(u).
\]
For $x\ne x'$, the polynomial $p_x-p_{x'}$ is nonzero and has total degree at most $\ell$. Hence, for uniform $u\in\F_{2^\lambda}^{\ell}$,
\[
 \Pr_u[F(x,u)=F(x',u)]\leq\ell2^{-\lambda}.
\]

Define the hash family $h_w(z):=\pi_{m_0}(wz)$. For $z\ne z'$, $W(z-z')$ is uniform on $\F_{2^\lambda}$. Therefore $\pi_{m_0}(W(z-z'))$ is zero with probability $2^{-m_0}$, and the family $\{h_w\}$ is two-universal.

Apply Lemma~\ref{thm:collision} to $Z_i=h_W(F(X,Y_{\mathrm{in}}))$ with source $X$, seed $Y_{\mathrm{in}}$, hash seed $W$, side information $A_iE$ and condition $\hmin(X\mid A_iE)_\sigma\geq k-\lambda\varrho M$. For output length $m_0$ and $\xi:=\ell2^{-\lambda}$.  Lemma~\ref{thm:collision} gives
\[
 \dist\!\left(\sigma_{Z_iA_iYWE},
       \tau_{m_0}\otimes\sigma_{A_iYWE}\right)\leq\varepsilon_0.
\]

For each $j<i$,  by construction,
\[
 q_j(Y|_{S_i\cap S_j})=p_X(Y|_{S_j}).
\]
Recall that $A_i$ contain the coefficients of $q_j$ for all $j<i$. So $A_i,Y,W$ determine  $Z_j$ for every $j<i$ through
\[
 Z_j=\pi_{m_0}\!\left(Wq_j(Y|_{S_i\cap S_j})\right).
\]
Let $\mathcal R_i$ be the deterministic channel that uses this
equation to compute $Z_{<i}$ from $A_i,Y,W$, then discards $A_i$,
and leaves $Z_i,Y,W,E$ unchanged. Thus
\begin{align*}
 \mathcal R_i(\sigma_{Z_iA_iYWE})
 &=\sigma_{Z_iZ_{<i}YWE},\\
 \mathcal R_i(\tau_{m_0}\otimes\sigma_{A_iYWE})
 &=\tau_{m_0}\otimes\sigma_{Z_{<i}YWE}.
\end{align*}
Contractivity gives
\[
 \dist\!\left(\sigma_{Z_iZ_{<i}YWE},
       \tau_{m_0}\otimes\sigma_{Z_{<i}YWE}\right)\leq\varepsilon_0.
\]

Averaging over $y_{\mathrm{out}}$ gives $\dist(\rho^{(i)},\rho^{(i-1)})\leq\varepsilon_0$.

\emph{Total error.} Summing the adjacent distances yields
\[
 \dist(\rho_{ZYWE},\tau_m\otimes\rho_{YWE})
 \leq\sum_{i=1}^M\dist(\rho^{(i)},\rho^{(i-1)})
 \leq M\varepsilon_0.
\]
By~\eqref{eq:design-error}, $M\varepsilon_0\leq\eta$, proving the strong extraction guarantee.
\end{proof}

\subsection{Linear output at constant entropy rate}\label{sec:long-output}
\begin{lemma}\label{lem:long-output}
Fix a constant $0<\mu\leq1$. There are positive constants $c_\mu,a_\mu,A_\mu$ such that, for all sufficiently large $n$, every $\mu n\leq k\leq n$, and every $0<\eta\leq1/2$ satisfying $h=\log(n/\eta)\leq c_\mu n$, there is an explicit strong quantum-proof $(k,\eta)$ extractor $\Ext:\{0,1\}^n\times\{0,1\}^d\to\{0,1\}^m$ such that
\[
 d\leq A_\mu h\log^2(n/h),\qquad m\geq a_\mu n.
\]
The map $\Ext$ is computable in polynomial time for fixed $\mu$.
\end{lemma}
\begin{proof}
In Theorem~\ref{thm:design}, take
\[
 \lambda=\lceil16h\rceil,\qquad
 \ell=\lceil\log(n/\lambda)\rceil,\qquad
 m_0=\lfloor\lambda/4\rfloor,\qquad
 \varrho=2,\qquad
 M=\left\lfloor\frac{\mu n}{8\lambda}\right\rfloor.
\]
Choose $c_\mu$ small enough that $\lambda\leq\mu n/16$ for all sufficiently large $n$. Then $\ell\geq1$, $\lambda2^\ell=O(n)$, $M\geq\mu n/(16\lambda)$, and $m_0\geq\lambda/8$. Thus the output length is $m=Mm_0\geq\mu n/128$. The weak design has $v=O(\ell^2)$, so the seed length is
$d=\lambda(v+1)=O_\mu(h\log^2(n/h))$.

Since $k-\lambda\varrho M\geq3\mu n/4$, $m_0\leq\lambda/4$,
and $M\leq n$, the left-hand side of~\eqref{eq:design-error}
is at most
\[
 \frac n2\left(2^{\lambda/8-3\mu n/8}
       +\sqrt{\ell}\,2^{-3\lambda/8}\right).
\]
Using $\lambda\leq\mu n/16$, $h\leq\mu n/256$, $\lambda\geq16h$, $\ell\leq h$, and $\eta=n2^{-h}$,
\begin{align*}
 \frac n2\,2^{\lambda/8-3\mu n/8}
 &\leq\frac n2\,2^{-47\mu n/128}
 =\frac\eta2\,2^{h-47\mu n/128}\leq\frac\eta2,\\
 \frac n2\sqrt\ell\,2^{-3\lambda/8}
 &\leq\frac\eta2\sqrt\ell\,2^{-5h}\leq\frac\eta2.
\end{align*}

Shoup's algorithm~\cite[Theorem~4.1]{Shoup90} deterministically constructs an irreducible polynomial of degree $\lambda$ over $\F_2$. The running time is polynomial in $\lambda$. Use the irreducible polynomial to represent $\F_{2^\lambda}$ as a quotient of $\F_2[z]$. Field arithmetic then takes time polynomial in $\lambda$. The weak-design construction and each polynomial evaluation therefore take time polynomial in $n$.
\end{proof}

%% file: seed_generation.tex
\section{Reducing the Seed Length}\label{sec:seed-generation}
Put $h=\log(n/\eta)$. For sources of constant entropy rate, we reduce the seed length in Lemma~\ref{lem:long-output} from $O_\mu(h\log^2(n/h))$ to $O_\mu(h)$. The output length remains linear. %The construction allows $h\leq c_\mu n$ for a sufficiently small constant $c_\mu>0$. In particular, the error may be exponentially small.

%The construction follows the classical block-sampling and extraction framework~\cite{NZ96,Zuc97}. Chung et al.~\cite[Section~4]{CCVW} studied a related approach with quantum side information in a manuscript that was later withdrawn.\footnote{The authors withdrew the manuscript because an error in its condenser argument invalidates the main result.}

%A public seed $A$ selects a small subset of the coordinates of $X$. Let $B$ contain the selected bits and $C$ the remaining bits. We use an extractor $G$ to generate a seed from $C$, then apply the multilinear extractor $F$ to $B$ with that seed:
%\[
% \Ext(X;A,Y)=F\bigl(B,G(C,Y)\bigr).
%\]
%The seed $Y$ of $G$ is independent of $XAE$. The pair $(A,Y)$ is the external seed. To serve as the seed for $F$, the output $G(C,Y)$ must be close to uniform jointly with $YBAE$.

\subsection{Sampling and extraction tools}
Combining Vadhan's sampler~\cite[Lemma~8.4]{Vad04} with K\"onig and Renner's entropy-sampling theorem~\cite[Corollary~6.19]{KR} gives an explicit entropy sampler for sources with quantum side information.

\begin{lemma}[Entropy sampling]\label{lem:public-entropy-sampling}
For every $0<\varepsilon<1$, there is a constant
$C_\varepsilon>0$ such that the following holds.
Let $N,w\geq1$ be integers, let $0<\delta\leq1/2$, and let
$C_\varepsilon\log(1/\delta)\leq t\leq N$ be an integer.
There is an explicit map
\[
 S:\{0,1\}^{r}\to\binom{[N]}t,
 \qquad
 r=O_\varepsilon\!\left(\log N+\log(1/\delta)\right),
\]
with the following entropy guarantee.

Let $\rho_{XE}$ be a normalized cq state with
$X\in(\{0,1\}^w)^N$ and
$\hmin(X\mid E)_\rho\geq\mu Nw$, where $0<\mu\leq1$.
For a uniform seed $R$ independent of $XE$, set
\[
 B=X_{S(R)},
 \qquad
 \kappa=\frac{N}{tw}.
\]
If $\kappa\leq0.15$, then
\begin{equation}\label{eq:public-entropy-sampling}
 \hmin^\tau(B\mid RE)_\rho
 \geq
 \bigl(\mu-3\varepsilon-2\kappa\log(1/\kappa)\bigr)wt
\end{equation}
for every $0<\tau<1$ satisfying
\[
 4\cdot2^{-\varepsilon Nw}+6\delta^{1/4}\leq\tau^2.
\]
\end{lemma}

The short-seed hash family of~\cite[Theorem~10]{LHL} gives the following extractor.

\begin{lemma}[Short-seed extraction]\label{lem:short-hash}
There is an absolute constant $C>0$ such that, for every integer $n\geq1$, entropy threshold $0\leq k\leq n$, error $0<\eta\leq1/2$, and integer output length $1\leq m\leq k-C\log(1/\eta)$, there is an explicit strong quantum-proof $(k,\eta)$ extractor
\[
 H:\{0,1\}^n\times\{0,1\}^d\to\{0,1\}^m,
 \qquad d\leq2m+C\log(n/\eta).
\]
The map $H$ is computable in time polynomial in $n+m+\log(1/\eta)$.
\end{lemma}

The original statements and the derivations of Lemmas~\ref{lem:public-entropy-sampling} and~\ref{lem:short-hash} appear in Appendix~\ref{app:sampling-extraction-tools}.

\subsection{Generating the internal seed}\label{sec:internal-seed}\label{sec:internal-seed-analysis}
\begin{theorem}\label{thm:short-seed-generator}
Fix constants $0<\mu\leq1$, $A>0$, and $C_M>0$. There is a constant $c=c(\mu,A,C_M)>0$ such that, for all sufficiently large $n$, every $\mu n\leq k\leq n$, and every $2^{-cn}\leq\eta\leq1/2$, there is an explicit strong quantum-proof $(k,\eta)$ extractor $G:\{0,1\}^n\times\{0,1\}^d\to\{0,1\}^M$ with
\[
 d=O_{\mu,A,C_M}(h),\qquad M=\left\lfloor C_Mh\log^A(n/h)\right\rfloor,
 \qquad h=\log(n/\eta).
\]
The map $G$ is computable in polynomial time for fixed $\mu,A,C_M$.
\end{theorem}

\paragraph{Construction.}
Set $\nu=\mu/4$ and fix an integer $Q\geq1+2/\nu$. The constants $C_0,C_w,C_t$ and the sampler accuracy $\varepsilon$ depend only on $\mu,A,C_M$. Their values are fixed at the start of the proof of Theorem~\ref{thm:short-seed-generator}.
\begin{enumerate}
\item \textbf{String lengths.} Let
\[
 d_0=\lceil C_0h\rceil,\qquad
 m_j=\lfloor d_j/4\rfloor,\qquad d_{j+1}=d_j+m_j,
\]
and stop at the first $T$ with $d_T-d_0\geq M$. Start with an independent uniform $d_0$-bit string $Z_0$.
\item \textbf{Samples of increasing length.} Pad $X$ with zeroes to $N=\lceil n/w\rceil$ symbols of $w=\lceil C_w\sqrt{n/h}\rceil$ bits. Let $t_0=\lceil C_th\rceil$ and $t_i=Q^it_0$ for $0\leq i<T$. For each $i$, choose a sampler from Lemma~\ref{lem:public-entropy-sampling} with sample size $t_i$, accuracy $\varepsilon$, and failure probability
\[
 \delta=\left(\frac{\eta}{32T}\right)^8.
\]
Let $r$ be the maximum of the sampler seed lengths. Choose a uniform $R\in\{0,1\}^{r}$ independent of $XE,Z_0$. The sampler for sample size $t_i$ reads the prefix of $R$ of the required length. Write $S_i(R)$ for the resulting set of coordinates. Let $B_i=X_{S_i(R)}$ list the sampled symbols in increasing coordinate order. Set $B_{<i}=(B_0,\ldots,B_{i-1})$, with $B_{<0}$ empty.
\item \textbf{Extraction in reverse order.} At step $j=0,\ldots,T-1$, process $B_{T-1-j}$. Put $\eta_0=\eta/(2T)$. Choose a strong quantum-proof extractor $H_j$ from Lemma~\ref{lem:short-hash}. The input length is $wt_{T-1-j}$, the output length is $m_j$, and the error is $\eta_0$. The entropy threshold of $H_j$ is $m_j+C\log(1/\eta_0)$. Extend $H_j$ to accept a $d_j$-bit seed. The map $H_j$ reads the required prefix and ignores the remaining seed bits. Append $H_j(B_{T-1-j},Z_j)$ to $Z_j$:
\begin{equation}\label{eq:reverse-pool-construction}
 Z_{j+1}=Z_j\circ H_j(B_{T-1-j},Z_j).
\end{equation}
The external seed is $(R,Z_0)$. Output any fixed $M$ of the appended bits.
\end{enumerate}
The geometric growth of the sample lengths lets $B_i$ retain entropy after $B_{<i}$ is revealed. We process the sampled blocks in the order $B_{T-1},\ldots,B_0$. After processing $B_i$, the unprocessed samples are exactly $B_{<i}$.

\begin{proof}[Proof of Theorem~\ref{thm:short-seed-generator}]
\emph{Parameter choices.} Set $\nu=\mu/4$, $\varepsilon=\mu/100$, and $Q=\lceil1+8/\mu\rceil$. Let $C_\varepsilon\geq1$ and $C$ be the constants in Lemmas~\ref{lem:public-entropy-sampling} and~\ref{lem:short-hash}, respectively. Choose $C_w\geq4$ and set $C_t,C_0$ so that
\begin{equation}\label{eq:seed-constant-choices}
\begin{gathered}
 3\varepsilon+\frac4{C_w^2}\log(C_w^2/2)\leq\mu/4,\\
 C_t\geq48C_\varepsilon,\qquad
 C_0\geq20,\quad C_0\geq8C.
\end{gathered}
\end{equation}
Choose $0<c\leq\varepsilon/8$ so that, whenever $h\leq2cn$,
\begin{equation}\label{eq:seed-size-requirements}
 t_{T-1}\leq N,\qquad
 \frac{d_T}{4}+4Ch+1\leq\frac\nu2wt_0.
\end{equation}
Put $\eta_0=\eta/(2T)$ and $\tau=\eta/(8T)$.

To see that such a constant $c$ exists, compare the sample and seed lengths. Since $d_j\geq C_0h\geq20$, we have $m_j=\lfloor d_j/4\rfloor\geq d_j/4-1\geq d_j/5$, and hence $d_{j+1}\geq(6/5)d_j$. The stopping rule gives $T=O(\log\log(n/h))$ and $d_T\leq\frac54(d_0+M)$. Thus
\begin{equation}\label{eq:bootstrap-growth}
 t_{T-1}=h(\log(n/h))^{O(1)},\qquad
 d_T=O(h\log^A(n/h)),\qquad
 N,wt_0=\Theta(\sqrt{nh}).
\end{equation}
Both conditions in~\eqref{eq:seed-size-requirements} hold when $n/h$ is sufficiently large, because $\sqrt{n/h}$ grows faster than either power of $\log(n/h)$. This is ensured by taking $c$ small enough. For sufficiently large $n$, the stated error range gives $h\leq2cn$.

\emph{Entropy of the sampled blocks.} Fix a cq source $\rho_{XE}$ with $\hmin(X\mid E)\geq\mu n$. We first prove $\hmin^\tau(B_i\mid RE)_\rho\geq\nu wt_i$ by Lemma~\ref{lem:public-entropy-sampling}. We then account for revealing $B_{<i}$.

Padding preserves the source entropy. Since $n'=Nw<n+w=n+O(\sqrt{n/h})\leq2n$ for sufficiently large $n$, the padded source has conditional entropy rate at least $\mu/2$. Also, $T=O(\log\log n)\leq n$ gives
\begin{equation}\label{eq:seed-error-logarithms}
 \log(1/\delta)=8\log(32T/\eta)\leq48h,\qquad
 \log(n'/\eta_0)\leq\log(4nT/\eta)\leq4h.
\end{equation}
We check the sample-size condition using $C_t\geq48C_\varepsilon$ from~\eqref{eq:seed-constant-choices}, the logarithmic bound~\eqref{eq:seed-error-logarithms}, and $t_{T-1}\leq N$ from~\eqref{eq:seed-size-requirements}:
\begin{equation}\label{eq:seed-sampling-conditions}
 C_\varepsilon\log(1/\delta)\leq48C_\varepsilon h
 \leq C_th\leq t_0\leq t_i\leq t_{T-1}\leq N.
\end{equation}
Next, $t_i\geq t_0\geq h$ and $w^2\geq C_w^2n/h$, together with $n'\leq2n$, give
\begin{equation}\label{eq:seed-kappa-bound}
 \kappa_i:=\frac{N}{t_iw}=\frac{n'}{t_iw^2}
 \leq\frac{2n}{h\cdot C_w^2n/h}
 =\frac{2}{C_w^2}\leq\frac18.
\end{equation}
Since $\kappa\log(1/\kappa)$ is increasing on $(0,1/8]$, equations~\eqref{eq:seed-kappa-bound} and~\eqref{eq:seed-constant-choices} imply
\begin{equation}\label{eq:seed-sampling-entropy-loss}
 3\varepsilon+2\kappa_i\log(1/\kappa_i)
 \leq3\varepsilon+\frac4{C_w^2}\log(C_w^2/2)
 \leq\mu/4.
\end{equation}

For the smoothing condition, $T\leq n$ ensures $512T^2 2^{-3\varepsilon n/4}\leq1$ for sufficiently large $n$. Using $n'\geq n$, $\eta\geq2^{-cn}$, and $c\leq\varepsilon/8$, we obtain
\begin{equation}\label{eq:seed-smoothing-terms}
\begin{aligned}
 2^{-\varepsilon n'}
 &\leq\eta^2 2^{-(\varepsilon-2c)n}
 \leq\eta^2 2^{-3\varepsilon n/4}
 \leq\frac{\eta^2}{512T^2},\\
 6\delta^{1/4}&=6\left(\frac{\eta}{32T}\right)^2
 =\frac{3\eta^2}{512T^2}.
\end{aligned}
\end{equation}
Substituting~\eqref{eq:seed-smoothing-terms} gives
\begin{equation}\label{eq:seed-smoothing-condition}
 4\cdot2^{-\varepsilon n'}+6\delta^{1/4}
 \leq\frac{4\eta^2+3\eta^2}{512T^2}
 \leq\frac{\eta^2}{64T^2}=\tau^2.
\end{equation}
Let $\rho_{B_0\cdots B_{T-1}RE}$ be the joint cq state obtained
from $\rho_{XE}\otimes\tau_r$ by computing
$B_q=X_{S_q(R)}$ for $0\leq q<T$ and discarding $X$.

Equations~\eqref{eq:seed-sampling-conditions},
\eqref{eq:seed-kappa-bound}, and~\eqref{eq:seed-smoothing-condition}
verify the hypotheses of Lemma~\ref{lem:public-entropy-sampling}
with entropy rate $\mu/2$, sample size $t_i$, accuracy $\varepsilon$,
failure probability $\delta$, and smoothing parameter $\tau$.
The lemma bounds the entropy of $B_i$ given $E$ and the seed prefix
read by $S_i$.
Revealing the unused suffix of $R$ does not change the bound,
because the suffix is independent of $B_i$, $E$, and the prefix.
Using~\eqref{eq:seed-sampling-entropy-loss} and $\nu=\mu/4$, we obtain
\begin{equation}\label{eq:sampled-block-before-leakage}
 \hmin^\tau(B_i\mid RE)_\rho
 \geq\left(\frac\mu2-3\varepsilon
          -2\kappa_i\log(1/\kappa_i)\right)wt_i
 \geq\nu wt_i.
\end{equation}

Apply Lemma~\ref{lem:smooth-classical-leakage} with source $B_i$,
leaked register $B_{<i}$, and side information $RE$.
The string $B_{<i}$ has $w\sum_{q<i}t_q$ bits, so
\begin{equation}\label{eq:sampled-prefix-leakage}
 \hmin^\tau(B_i\mid B_{<i}RE)_\rho
 \geq\hmin^\tau(B_i\mid RE)_\rho-w\sum_{q<i}t_q.
\end{equation}
Substituting~\eqref{eq:sampled-block-before-leakage} and
$t_q=Q^qt_0$ into~\eqref{eq:sampled-prefix-leakage}, and using
$Q\geq1+2/\nu$, gives
\begin{align}\label{eq:sampled-block-entropy}
 \hmin^\tau(B_i\mid B_{<i}RE)_\rho
 &\geq\nu wt_i-w\sum_{q<i}t_q\notag\\
 &\geq\left(\nu-\frac1{Q-1}\right)wt_i
 \geq\frac\nu2wt_i.
\end{align}
In particular, every $B_i$ has $\Omega(\sqrt{nh})$ smooth
conditional min-entropy given $B_{<i}RE$.

\emph{Requirements for extraction.} Fix $0\leq j<T$ and put $i=T-1-j$. Apply Lemma~\ref{lem:short-hash} with input length $wt_i$, output length $m_j$, and error $\eta_0$. The resulting extractor $H_j$ needs at most $2m_j+C\log(wt_i/\eta_0)$ seed bits. Since $wt_i\leq Nw=n'$ and $m_j\leq d_j/4$, equation~\eqref{eq:seed-error-logarithms} gives
\[
 2m_j+C\log(wt_i/\eta_0)
 \leq d_j/2+C\log(n'/\eta_0)
 \leq d_j/2+4Ch\leq d_j.
\]
The last inequality uses $4Ch\leq C_0h/2\leq d_j/2$. Thus $Z_j$ supplies enough seed bits for $H_j$. For the smooth entropy requirement, equations~\eqref{eq:seed-error-logarithms} and~\eqref{eq:seed-size-requirements} give
\begin{equation}\label{eq:seed-extraction-requirements}
 m_j+C\log(1/\eta_0)+1
 \leq\frac{d_T}{4}+4Ch+1
 \leq\frac\nu2wt_0\leq\frac\nu2wt_i.
\end{equation}
Together, equations~\eqref{eq:sampled-block-entropy} and~\eqref{eq:seed-extraction-requirements} allow us to apply Lemma~\ref{lem:smooth-extraction} to each $H_j$.

\emph{Hybrid states.} Let $\rho$ be the joint state produced by the construction. For $0\leq j\leq T$, replace $Z_j$ by an independent uniform $d_j$-bit string. Starting with
\[
 \tau_{d_j}\otimes\rho_{B_{<T-j}RE}
\]
apply~\eqref{eq:reverse-pool-construction} for $j,\ldots,T-1$, discarding each sampled block after use. Denote the resulting state by $\rho^{(j)}_{Z_TRE}$. Since $Z_0$ is independent and uniform, the endpoints are
\[
 \rho^{(0)}_{Z_TRE}=\rho_{Z_TRE},\qquad
 \rho^{(T)}_{Z_TRE}=\tau_{d_T}\otimes\rho_{RE}.
\]
We will prove $\dist(\rho^{(j)},\rho^{(j+1)})\leq\eta_0+4\tau$ for every $0\leq j<T$.

\emph{Distance between consecutive hybrids.} Fix $0\leq j<T$ and put $i=T-1-j$. Let $V\sim U_{d_j}$ be independent of $B_iB_{<i}RE$. To obtain $\rho^{(j)}$, first compute the $d_{j+1}$-bit string
$V\circ H_j(B_i,V)$ in step $j$.
To obtain $\rho^{(j+1)}$, replace $V\circ H_j(B_i,V)$
by a uniform $d_{j+1}$-bit string independent of $B_{<i}RE$.
Both constructions then apply~\eqref{eq:reverse-pool-construction}
for steps $j+1,\ldots,T-1$, using the remaining blocks $B_{<i}$,
and retain $Z_TRE$.

The map $H_j$ from Lemma~\ref{lem:short-hash} is a strong quantum-proof $(m_j+C\log(1/\eta_0),\eta_0)$ extractor. Apply Lemma~\ref{lem:smooth-extraction} to $H_j$. Take source $B_i$, seed $V$, side information $B_{<i}RE$, and smoothing parameter $\tau$. Equations~\eqref{eq:sampled-block-entropy} and~\eqref{eq:seed-extraction-requirements} supply the required smooth entropy. The strong guarantee of $H_j$ gives near-uniform output even given the seed prefix of $V$. Therefore
\[
 \dist\!\left(\rho_{(V\circ H_j(B_i,V))B_{<i}RE},
       \tau_{d_{j+1}}\otimes\rho_{B_{<i}RE}\right)
 \leq \eta_0+4\tau.
\]
Let $\Phi_j$ be a channel with input registers $Z_{j+1}B_{<i}RE$.
The channel applies steps $j+1,\ldots,T-1$
of~\eqref{eq:reverse-pool-construction} to extend $Z_{j+1}$ to $Z_T$
using the remaining blocks $B_{<i}$.
It then discards $B_{<i}$ and outputs $Z_TRE$, leaving $RE$ unchanged. Applying $\Phi_j$ to $\rho_{(V\circ H_j(B_i,V))B_{<i}RE}$ gives $\rho^{(j)}$. Applying $\Phi_j$ to $\tau_{d_{j+1}}\otimes\rho_{B_{<i}RE}$ gives $\rho^{(j+1)}$. Contractivity yields
\begin{equation}\label{eq:reverse-pool-hybrid}
 \dist(\rho^{(j)},\rho^{(j+1)})\leq\eta_0+4\tau.
\end{equation}

\emph{Total error.} Summing~\eqref{eq:reverse-pool-hybrid} over $0\leq j<T$ yields
\[
 \dist(\rho_{Z_TRE},\tau_{d_T}\otimes\rho_{RE})
 \leq\sum_{j=0}^{T-1}\dist(\rho^{(j)},\rho^{(j+1)})
 \leq T(\eta_0+4\tau)=\eta.
\]
The string $Z_T$ contains the unchanged initial string $Z_0$. If $O$ is the selected $M$-bit output, discarding the other appended bits gives
\[
 \dist\!\left(\rho_{OZ_0RE},
       \tau_M\otimes\rho_{Z_0RE}\right)\leq\eta,
 \qquad
 \rho_{Z_0RE}=\tau_{d_0}\otimes\rho_{RE}.
\]
Thus $O$ satisfies the strong guarantee for the external seed $(R,Z_0)$.

By Lemma~\ref{lem:public-entropy-sampling} and~\eqref{eq:seed-error-logarithms}, the common seed $R$ has length $r=O(h)$. The external seed therefore has length $r+d_0=O(h)$. All sampling and hashing steps run in polynomial time.
\end{proof}

\subsection{Combining the extractors}
\begin{proof}[Proof of Theorem~\ref{thm:constant-rate}]
\emph{Construction.} Put $h=\log(n/\eta)$, $\tau=\eta/100$, $\varepsilon=\mu/100$, and $\delta=(\tau/100)^8$. Choose a sufficiently large constant integer $w=w(\mu)$. Put $N=\lfloor n/w\rfloor$, $n'=wN$, and $t=\lfloor\mu N/16\rfloor$. Partition the first $n'$ coordinates into $N$ consecutive blocks of $w$ bits. Choose the sampling map $T:\{0,1\}^r\to\binom{[N]}t$ from Lemma~\ref{lem:public-entropy-sampling} with parameters $\varepsilon,\delta$, and define
\[
 S(a)=\bigcup_{i\in T(a)}\{(i-1)w+1,\ldots,iw\}.
\]
For a source string $x$ and a seed value $a$, write
\[
 b=x_{S(a)},\qquad c=x_{[n]\setminus S(a)},
\]
with coordinates in increasing order. Set $n_B=wt$ and $n_C=n-n_B$. Then $n_B,n_C=\Theta_\mu(n)$, $n_B\leq\mu n/16$, and $r=O_\mu(h)$.

Apply Lemma~\ref{lem:long-output} at input length $n_B$, entropy rate $\mu/2$, and error $\eta/4$. We obtain a strong quantum-proof extractor
\[
 F:\{0,1\}^{n_B}\times\{0,1\}^{d_1}\to\{0,1\}^m,\qquad
 d_1=O_\mu\bigl(h\log^2(n/h)\bigr),\qquad m=\Omega_\mu(n).
\]
To generate the $d_1$-bit seed of $F$ from $c$, set $h_C=\log(4n_C/\eta)$. We have $h_C=\Theta_\mu(h)$. Also, $n_C/h_C=\Theta_\mu(n/h)$. Choose a sufficiently large constant $C_M$ depending only on $\mu$. Then $d_1\leq C_Mh_C\log^2(n_C/h_C)$. Apply Theorem~\ref{thm:short-seed-generator} at input length $n_C$, entropy rate $\mu/2$, error $\eta/4$, and exponent $A=2$. Truncating the output to $d_1$ bits gives a strong quantum-proof extractor
\[
 G:\{0,1\}^{n_C}\times\{0,1\}^{d_2}\to\{0,1\}^{d_1},\qquad
 d_2=O_\mu(h).
\]
Take $c_\mu$ sufficiently small in $h\leq c_\mu n$ and take $n$ sufficiently large. The parameter conditions for $S,F,G$ then hold. The public seed is $(a,y)\in\{0,1\}^{r+d_2}$. Define
\[
 \Ext(x;a,y)=F\bigl(b,G(c,y)\bigr).
\]

\emph{Correctness.} Fix a normalized cq state $\rho_{XE}$ with $\hmin(X\mid E)\geq\mu n$. Take independent uniform seeds $A\sim U_r$ and $Y\sim U_{d_2}$, both independent of $XE$. Let $B=X_{S(A)}$ and $C=X_{[n]\setminus S(A)}$.

We first check that $B$ has enough smooth min-entropy for $F$. Omitting the last $n-n'$ bits leaves $\hmin(X_{[n']}\mid E)\geq\mu n-w\geq(7\mu/8)n'$ by~\eqref{eq:classical-pair-chain}. Since $t\geq\mu N/32$, we have $\kappa=N/(tw)\leq32/(\mu w)$. Choose $w=w(\mu)$ large enough that
\[
 \kappa\leq0.15,\qquad
 3\varepsilon+2\kappa\log(1/\kappa)\leq\mu/8.
\]

The remaining conditions of Lemma~\ref{lem:public-entropy-sampling} are
\[
 t\geq C_\varepsilon\log(1/\delta),\qquad
 4\cdot2^{-\varepsilon n'}+6\delta^{1/4}\leq\tau^2.
\]
After fixing $w$, choose
\[
 c_\mu\leq
 \min\left\{\frac{\mu}{512C_\varepsilon w},
             \frac{\varepsilon}{4}\right\}.
\]
For sufficiently large $n$, we have $n\geq100$,
$n'\geq n/2$, and
$t\geq\mu N/32\geq\mu n/(64w)$.
Since $\delta=(\tau/100)^8$ and
$\log(n/\tau)\leq c_\mu n$,
\[
 C_\varepsilon\log(1/\delta)
 =8C_\varepsilon\log(100/\tau)
 \leq8C_\varepsilon c_\mu n
 \leq\frac{\mu n}{64w}
 \leq t.
\]
Moreover,
\[
 \varepsilon n'
 \geq\frac{\varepsilon n}{2}
 \geq2c_\mu n
 \geq2\log(100/\tau).
\]
Consequently,
\[
 4\cdot2^{-\varepsilon n'}+6\delta^{1/4}
 \leq4(\tau/100)^2+6(\tau/100)^2
 =10(\tau/100)^2
 <\tau^2.
\]

Applying the lemma at entropy rate $7\mu/8$ gives
\begin{equation}\label{eq:composition-sampled-entropy}
 \hmin^\tau(B\mid AE)\geq\frac{3\mu}{4}n_B
 \geq\frac\mu2n_B+1.
\end{equation}

Next, we check that $C$ has enough conditional min-entropy for $G$, even given $B$. For each fixed $A$, the pair $(B,C)$ is a permutation of $X$. Independence of $A$ and~\eqref{eq:classical-pair-chain} therefore give
\begin{equation}\label{eq:composition-complement-entropy}
 \hmin(C\mid BAE)\geq\hmin(X\mid E)-n_B
 \geq\mu n-n_B\geq\frac\mu2n_C.
\end{equation}

We now replace the internal seed $G(C,Y)$ by an independent $V\sim U_{d_1}$. Equation~\eqref{eq:composition-complement-entropy} and the strong guarantee of $G$ give
\[
 \dist\!\left(\rho_{G(C,Y)BAYE},
       \tau_{d_1}\otimes\rho_{BAYE}\right)\leq\eta/4.
\]
Applying $F$ to $B$ and the internal seed, and discarding both inputs, yields
\begin{equation}\label{eq:composition-seed-replacement}
 \dist\!\left(\rho_{\Ext(X;A,Y)AYE},\rho_{F(B,V)AYE}\right)\leq\eta/4.
\end{equation}
Finally, $Y$ is independent of $BAE$, so~\eqref{eq:composition-sampled-entropy} also holds with side information $AYE$. Lemma~\ref{lem:smooth-extraction} applied to $F$ with source $B$ and seed $V$ gives
\begin{equation}\label{eq:composition-uniform-seed}
 \dist\!\left(\rho_{F(B,V)AYE},\tau_m\otimes\rho_{AYE}\right)
 \leq\eta/4+4\tau.
\end{equation}
Combining~\eqref{eq:composition-seed-replacement} and~\eqref{eq:composition-uniform-seed}, we obtain
\[
 \dist\!\left(\rho_{\Ext(X;A,Y)AYE},\tau_m\otimes\rho_{AYE}\right)
 \leq\eta/2+4\tau<\eta.
\]
The public seed has length $r+d_2=O_\mu(h)$, and the output has length $m=\Omega_\mu(n)$. The maps $S,F,G$ are computable in polynomial time.
\end{proof}

\begin{corollary}\label{cor:entropy-fraction}
Fix constants $0<\mu,\theta<1$. There is a constant $c_{\mu,\theta}>0$ such that, for all sufficiently large $n$, every $\mu n\leq k\leq n$, and every $0<\eta\leq1/2$ satisfying $\log(n/\eta)\leq c_{\mu,\theta}n$, there is an explicit strong quantum-proof $(k,\eta)$ extractor $\Ext:\{0,1\}^n\times\{0,1\}^d\to\{0,1\}^m$ with
\[
 d=O_{\mu,\theta}(\log(n/\eta)),\qquad m=\lfloor(1-\theta)k\rfloor.
\]
The map $\Ext$ is computable in polynomial time for fixed $\mu,\theta$.
\end{corollary}
\begin{proof}
Put $m=\lfloor(1-\theta)k\rfloor$. At entropy rate $\theta\mu$, Theorem~\ref{thm:constant-rate} provides $b=\lfloor\beta n\rfloor$ output bits for some constant $\beta>0$ depending only on $\mu,\theta$. Set $J=\lceil m/b\rceil=O_{\mu,\theta}(1)$. Apply the theorem with error $\eta/J$ to obtain an extractor $F$ with output length $b$ and seed length $d_0=O_{\mu,\theta}(\log(n/\eta))$. Taking $c_{\mu,\theta}$ sufficiently small ensures that $\log(nJ/\eta)\leq c_{\theta\mu}n$, as required by the theorem.

For a seed $(y_1,\ldots,y_J)$, concatenate the outputs
\[
 F(x,y_1)\circ\cdots\circ F(x,y_J)
\]
and retain the first $m$ bits. This defines $\Ext$.

Fix a source  $\rho_{XE}$ with $\hmin(X\mid E)\geq k$, and let $Y_1,\ldots,Y_J$ be independent uniform seeds, independent of $XE$. Write $Z_i=F(X,Y_i)$. Before the $j$th call, the preceding outputs contain $(j-1)b<m$ bits. The classical leakage bound therefore gives
\[
 \hmin(X\mid Z_1\cdots Z_{j-1}Y_1\cdots Y_{j-1}E)_\rho
 \geq k-(j-1)b\geq k-m\geq\theta\mu n.
\]
Thus every call meets the entropy requirement of $F$. The sequential-composition lemma~\cite[Lemma~A.4]{DPVR} gives the strong guarantee for $\Ext$ with error at most $J(\eta/J)=\eta$. Its seed length is $Jd_0=O_{\mu,\theta}(\log(n/\eta))$, and its running time is polynomial.
\end{proof}

%% file: high_entropy.tex
\section{Extraction with Small Entropy Deficiency}\label{sec:high-entropy}
Fix $\alpha>0$. We extract from sources of conditional min-entropy $n-a$ with seed length $(1+\alpha)\log a+(2+\alpha)\log(1/\eta)+O_\alpha(1)$. The output loss is $O_\alpha(a+\log(1/\eta))$ bits.

\subsection{A base extractor with small output loss}\label{sec:base-high-entropy}
We start with the classical extractor of Goldreich and Wigderson.

\begin{theorem}[Goldreich--Wigderson {\cite[Theorem~4.2 and Lemma~4.3]{GW97}}]\label{thm:classical-high-entropy}
For every integer $u\geq2$, integer entropy threshold $1\leq k<u$, and error $0<\varepsilon<1/4$, there is an explicit strong classical $(k,\varepsilon)$ extractor $H:\{0,1\}^u\times\{0,1\}^d\to\{0,1\}^m$ such that, writing $b=u-k$,
\[
 d=O(b+\log(1/\varepsilon)),\qquad
 m\geq u-O(b+\log(1/\varepsilon)).
\]
The map $H$ is computable in time polynomial in $u+\log(1/\varepsilon)$.
\end{theorem}

\begin{theorem}[Berta--Fawzi--Scholz {\cite[Theorem~II.4]{BFS}}]\label{thm:high-entropy-lifting}
Let $u\geq2$, $1\leq k_0\leq u-1$, and $0<\eta_0<1$. For an absolute constant $C_B\geq1$, every strong classical $(k_0,\eta_0)$ extractor $H:\{0,1\}^u\times\{0,1\}^d\to\{0,1\}^m$ is also a strong quantum-proof $(k_0+1,C_B2^{u-k_0}\eta_0)$ extractor.
\end{theorem}

Combining Theorems~\ref{thm:classical-high-entropy} and~\ref{thm:high-entropy-lifting} gives a quantum-proof extractor with seed length and output loss $O(a+\log(1/\eta))$. We shorten its seed to $O(\log a+\log(1/\eta))$ by using our constant-rate extractor on a short block of the source.

\begin{lemma}\label{lem:deficiency-base}
For every integer $n\geq1$, entropy threshold $0\leq k\leq n$, and error $0<\eta\leq1/2$, there is an explicit strong quantum-proof $(k,\eta)$ extractor $\Ext:\{0,1\}^n\times\{0,1\}^d\to\{0,1\}^m$ such that, writing $a=n-k$,
\[
 d=O(\log a+\log(1/\eta)),\qquad
 m\geq n-O(a+\log(1/\eta)).
\]
The map $\Ext$ is computable in polynomial time.
\end{lemma}
\begin{proof}
Put $\widehat\eta=\min\{\eta,1/8\}$. Choose sufficiently large absolute constants $C_F,C_H,C_0$, in order.

\emph{Construction.} Set $n_B=\lceil C_0(a+\log(1/\widehat\eta))\rceil$. If $n-n_B-C_F(a+\log(1/\widehat\eta))<1$, output the empty string. Otherwise put
\[
 L=n-n_B,\qquad m=\lfloor L-C_F(a+\log(1/\widehat\eta))\rfloor,\qquad
 d_H=\lceil C_H(a+\log(1/\widehat\eta))\rceil.
\]
Split the input $x$ into its first $n_B$ bits $x_B$ and its remaining $L$ bits $x_Z$. We extract from $x_B$ to generate a seed for extraction from $x_Z$.

To construct $H_L$, set
\[
 b=\lceil a\rceil+1,\qquad
 \varepsilon=\frac{\widehat\eta}{C_B2^{b+1}}.
\]
Apply Theorem~\ref{thm:classical-high-entropy} with input length $L$, entropy threshold $L-b$, and error $\varepsilon$ to obtain a strong classical $(L-b,\varepsilon)$ extractor. Since $b+\log(1/\varepsilon)=O(a+\log(1/\widehat\eta))$, sufficiently large $C_F,C_H$ ensure at least $m$ output bits and at most $d_H$ seed bits. The condition $L-b\geq1$ also holds. Truncate the output and pad the seed to obtain
\[
 H_L:\{0,1\}^L\times\{0,1\}^{d_H}\to\{0,1\}^m.
\]

By Theorem~\ref{thm:high-entropy-lifting}, $H_L$ is a strong quantum-proof $(L-b+1,\widehat\eta/2)$ extractor, since $C_B2^b\varepsilon=\widehat\eta/2$. Its entropy threshold is $L-b+1=L-\lceil a\rceil\leq L-a$.

Apply Theorem~\ref{thm:constant-rate} at length $n_B$, entropy rate $3/4$, and error $\widehat\eta/2$ to obtain a strong quantum-proof $(3n_B/4,\widehat\eta/2)$ extractor. For sufficiently large $C_0$, the theorem's hypotheses hold and the output has at least $d_H$ bits. Retain the first $d_H$ bits to obtain
\[
 F:\{0,1\}^{n_B}\times\{0,1\}^{d_B}\to\{0,1\}^{d_H},\qquad
 d_B=O(\log n_B+\log(1/\widehat\eta)).
\]
With public seed $y\in\{0,1\}^{d_B}$, define
\[
 \Ext(x;y)=H_L\bigl(x_Z,F(x_B,y)\bigr).
\]

\emph{Correctness.} Fix a normalized cq source $\rho_{BZE}$ with $\hmin(BZ\mid E)\geq n-a$. The registers $B$ and $Z$ have lengths $n_B$ and $L$, respectively. For sufficiently large $C_0$, the entropy inequalities~\eqref{eq:classical-pair-chain} give
\[
 \hmin(B\mid ZE)\geq n_B-a\geq3n_B/4,\qquad
 \hmin(Z\mid E)\geq L-a.
\]
Take $Y\sim U_{d_B}$ independently of $BZE$. Put $W=F(B,Y)$ and $V=H_L(Z,W)$, so $V$ is the actual output. Let $U\sim U_{d_H}$ be independent of $BZYE$, and put $V'=H_L(Z,U)$. We compare $V$ with $V'$ by replacing the generated seed $W$ with $U$.

The strong guarantee of $F$, applied to source $B$ with side information $ZE$, gives
\[
 \dist(\rho_{WZYE},\rho_{UZYE})\leq\widehat\eta/2,
 \qquad \rho_{UZYE}=\tau_{d_H}\otimes\rho_{ZYE}.
\]
Apply $H_L$ to $Z$ and the seed register in both states, then discard those two input registers. Contractivity gives
\begin{equation}\label{eq:base-seed-replacement}
 \dist(\rho_{VYE},\rho_{V'YE})\leq\widehat\eta/2.
\end{equation}

It remains to bound the distance of $V'$ from uniform. Independence of $Y$ gives $\hmin(Z\mid YE)=\hmin(Z\mid E)\geq L-a$. The seed $U$ is uniform and independent of $ZYE$. Applying $H_L$ to source $Z$ with side information $YE$, and discarding $U$, yields
\begin{equation}\label{eq:base-uniform-seed-output}
 \dist(\rho_{V'YE},\tau_m\otimes\rho_{YE})\leq\widehat\eta/2.
\end{equation}
Combining~\eqref{eq:base-seed-replacement} and~\eqref{eq:base-uniform-seed-output} by the triangle inequality gives
\begin{align*}
 \dist(\rho_{VYE},\tau_m\otimes\rho_{YE})
 &\leq\dist(\rho_{VYE},\rho_{V'YE})
      +\dist(\rho_{V'YE},\tau_m\otimes\rho_{YE})\\
 &\leq\widehat\eta\leq\eta.
\end{align*}
Since $\rho_{YE}=\tau_{d_B}\otimes\rho_E$, this proves the strong extraction guarantee.

The public seed has length $d_B=O(\log n_B+\log(1/\widehat\eta))=O(\log a+\log(1/\widehat\eta))$. The output loss is $n-m=n_B+\lceil C_F(a+\log(1/\widehat\eta))\rceil=O(a+\log(1/\widehat\eta))$. 
%In the empty-output case, $n=O(a+\log(1/\widehat\eta))$. 
%Choosing the absolute constant in the output-loss bound sufficiently large makes the lower bound on $m$ in Lemma~\ref{lem:deficiency-base} nonpositive.
Both $F$ and $H_L$ are computable in polynomial time.
\end{proof}

\subsection{Seed-length refinement by limited independence}\label{sec:seed-compression}
We next reduce the seed length of the extractor in Lemma~\ref{lem:deficiency-base}, at the cost of a larger output loss.

\begin{theorem}[Seed length versus output loss]\label{thm:moment-tradeoff}
For every integer $N\geq1$, entropy threshold $0\leq k\leq N$, error $0<\eta<1/4$, and integer $\ell\geq1$, there is an explicit strong quantum-proof $(k,\eta)$ extractor $\Ext:\{0,1\}^N\times\{0,1\}^d\to\{0,1\}^m$ such that, writing $a=N-k$,
\begin{align}
 d&\leq2\log(1/\eta)+\log\ell+\frac a\ell+5,\label{eq:moment-seed}\\
 m&\geq N-O\!\left(a+\log(1/\eta)+\ell[\log a+\log(1/\eta)+\log\ell]\right).
 \label{eq:moment-output}
\end{align}
The map $\Ext$ is computable in time polynomial in $N+\log(1/\eta)+\ell$.
\end{theorem}
\begin{proof}
\emph{Construction.} Let $d_0=O(\log a+\log(1/\eta))$ be an integer bounding the seed length in Lemma~\ref{lem:deficiency-base} at deficiency $a$ and error $\eta/2$, uniformly over the source length. We generate $t$ candidate $d_0$-bit seeds and use $d=\log t$ public bits to select one. Choose $t=2^d$ to be the least power of two at least $12\ell\eta^{-2}2^{a/\ell}$.

We generate the candidate seeds by evaluating a polynomial over $\F_{2^w}$. The field must contain $t$ distinct evaluation points, and each field element must supply $d_0$ bits. Thus we take $w$ to be the least power of two at least $\max\{d_0,d\}$. A polynomial of degree at most $2\ell-1$ has $2\ell$ coefficients, so we reserve $s=2\ell w$ source bits for its coefficients.

Put $n=N-s$. If $n<\max\{1,a\}$, output the empty string. Otherwise, write the input as $(x,z)\in\{0,1\}^n\times\{0,1\}^s$. Apply Lemma~\ref{lem:deficiency-base} to the $n$-bit block and pad the seed to length $d_0$ to obtain a strong quantum-proof $(n-a,\eta/2)$ extractor
\[
 \Ext_0:\{0,1\}^n\times\{0,1\}^{d_0}\to\{0,1\}^m.
\]

Following Xun and Zuckerman~\cite[Section~1.6.1]{XZ}, use $z$ as the coefficients of a polynomial. Fix a binary basis of $\F_{2^w}$ and write $z$ as $2\ell$ field elements $z_0,\ldots,z_{2\ell-1}$. Define
\[
 p_z(u)=\sum_{q=0}^{2\ell-1}z_q u^q.
\]
For each $j\in[t]$, let $u_j\in\F_{2^w}$ be the element whose binary coordinates encode $j-1$. Let $\pi_{d_0}:\F_{2^w}\to\{0,1\}^{d_0}$ be the binary linear projection onto the first $d_0$ basis coordinates. Set
\[
 G:\{0,1\}^s\times[t]\to\{0,1\}^{d_0},\qquad
 G(z,j)=\pi_{d_0}(p_z(u_j)).
\]
The public seed is a uniform $j\in[t]$. Define the extractor by
\[
 \Ext((x,z);j)=\Ext_0(x;G(z,j)).
\]

\emph{Correctness.} Since $2^w\geq t$, the evaluation points are distinct. For uniform $z\sim U_s$, the Vandermonde matrices show that any $2\ell$ evaluations of $p_z$ are independent and uniform. After applying $\pi_{d_0}$, the values $G(U_s,1),\ldots,G(U_s,t)$ remain $2\ell$-wise independent and uniform.

We prove a general error bound for $\Ext((x,z);j)=\Ext_0(x;G(z,j))$.

\begin{claim}
Let $\Ext_0:\{0,1\}^n\times\{0,1\}^{d_0}\to\{0,1\}^m$ be a weak quantum-proof $(n-a,\eta_0)$ extractor. Let $G:\{0,1\}^s\times[t]\to\{0,1\}^{d_0}$ be a map such that $G(U_s,1),\ldots,G(U_s,t)$ are $2\ell$-wise independent and uniform. Suppose $t\geq\ell$. Then the function $\Ext:\{0,1\}^{n+s}\times[t]\to\{0,1\}^m$, defined by $\Ext((x,z);j)=\Ext_0(x;G(z,j))$, is a strong quantum-proof $(n+s-a,\eta)$ extractor, where
\begin{equation}\label{eq:quantum-compression-error}
 \eta=\eta_0+2^{a/(2\ell)}\sqrt{\frac{3\ell}{t}}.
\end{equation}
\end{claim}
\begin{proof}[Proof of the claim]
Fix a normalized cq state $\rho_{XZE}$ with $\hmin(XZ\mid E)\geq n+s-a$. Write $\rho_E^{xz}$ for the subnormalized blocks of $\rho_{XZE}$. By~\eqref{eq:entropy-domination}, there is a normalized state $\omega_E$ such that
\begin{equation}\label{eq:compression-domination}
 \rho_E^{xz}\preceq2^{a-n-s}\omega_E
 \quad\text{for every }x,z.
\end{equation}
Summing over $z$ gives $\rho_E^x:=\sum_z\rho_E^{xz}\preceq2^{a-n}\omega_E$. Hence $\hmin(X\mid E)\geq n-a$.

Let $J$ be uniform on $[t]$ and independent of $XZE$. Let $\rho_{VJE}$ be the joint state of $V=\Ext_0(X;G(Z,J))$, $J$, and $E$. Define $\sigma_{VJE}$ by replacing $G(Z,J)$ with a uniform seed $R\sim U_{d_0}$ independent of $XZEJ$. Thus the output register $V$ contains $\Ext_0(X;R)$ in $\sigma_{VJE}$, and $J$ is independent of $VE$. The weak guarantee for $\Ext_0$ gives
\[
 \dist(\sigma_{VJE},\tau_m\otimes\tau_{[t]}\otimes\rho_E)\leq\eta_0.
\]

We next compare $\rho_{VJE}$ and $\sigma_{VJE}$. For each public seed $j$ and output $y$, fix a binary measurement $\{M_{j,y},I-M_{j,y}\}$ on $E$, with $0\preceq M_{j,y}\preceq I$. Accept on outcome $M_{j,y}$. For each fixed $x$, put
\begin{align}
 A_{x,j}(z)&=M_{j,\Ext_0(x;G(z,j))}
       -\E_{r\sim U_{d_0}}M_{j,\Ext_0(x;r)},\notag\\
 B_x(z)&=\frac1t\sum_{j=1}^t A_{x,j}(z).\label{eq:compression-test}
\end{align}
By~\eqref{eq:compression-test}, the difference between the acceptance probabilities of the measurements $\{M_{j,y},I-M_{j,y}\}$ on $\rho_{VJE}$ and $\sigma_{VJE}$ is
\[
 \sum_{x,z}\Tr\bigl(\rho_E^{xz}B_x(z)\bigr).
\]

Fix $x$. We first bound the $2\ell$th moment of $B_x(z)$ for uniform $z\sim U_s$. Under $z\sim U_s$, the matrices $A_{x,j}(z)$ are mean-zero, Hermitian, and $2\ell$-wise independent. Each matrix $A_{x,j}(z)$ has norm at most one. Expand $t^{2\ell}\E_z B_x(z)^{2\ell}$ into ordered products $A_{x,j_1}\cdots A_{x,j_{2\ell}}$. Suppose an index $j$ occurs exactly once in $A_{x,j_1}\cdots A_{x,j_{2\ell}}$. Condition on the other distinct matrices in the product. The conditional mean of $A_{x,j}$ is zero, so the product has expectation zero. Thus only tuples with at most $\ell$ distinct indices contribute. Each product has norm at most one. Consequently,
\begin{align}
 t^{2\ell}\norm{\E_z B_x(z)^{2\ell}}
 &\leq\sum_{q=1}^{\ell}\binom tq q^{2\ell}
 \leq\sum_{q=1}^{\ell}\frac{t^q\ell^{2\ell}}{q!}\notag\\
 &\leq t^\ell\ell^\ell\sum_{q=1}^{\ell}\frac{\ell^q}{q!}
 \leq(3\ell t)^\ell.\label{eq:moment-count}
\end{align}
The third inequality in~\eqref{eq:moment-count} uses $t\geq\ell$. Since $B_x(z)^{2\ell}\succeq0$, equation~\eqref{eq:moment-count} gives
\begin{equation}\label{eq:operator-moment}
 0\preceq\E_{z\sim U_s}B_x(z)^{2\ell}
 \preceq\left(\frac{3\ell}{t}\right)^\ell I.
\end{equation}

Since $\sum_{x,z}\Tr\rho_E^{xz}=1$, Jensen's inequality, followed by~\eqref{eq:compression-domination} and~\eqref{eq:operator-moment}, gives
\begin{align*}
 \left|\sum_{x,z}\Tr\bigl(\rho_E^{xz}B_x(z)\bigr)\right|
 &\leq\left(\sum_{x,z}\Tr\bigl(\rho_E^{xz}B_x(z)^{2\ell}\bigr)\right)^{1/(2\ell)}\\
 &\leq\left(2^a\E_{x\sim U_n,z\sim U_s}\Tr\bigl(\omega_E B_x(z)^{2\ell}\bigr)\right)^{1/(2\ell)}\\
 &\leq2^{a/(2\ell)}\sqrt{\frac{3\ell}{t}}.
\end{align*}
Maximizing over the matrices $0\preceq M_{j,y}\preceq I$ gives $\dist(\rho_{VJE},\sigma_{VJE})\leq2^{a/(2\ell)}\sqrt{3\ell/t}$. The triangle inequality then gives~\eqref{eq:quantum-compression-error}.
\end{proof}

For the constructed $\Ext_0$ and $G$, we have $\eta_0=\eta/2$ and $t\geq12\ell\eta^{-2}2^{a/\ell}$. Equation~\eqref{eq:quantum-compression-error} therefore bounds the total error by $\eta$.

\emph{Parameters.} Rounding $t$ up to a power of two gives
\[
 d<2\log(1/\eta)+\log\ell+\frac a\ell+\log12+1
 <2\log(1/\eta)+\log\ell+\frac a\ell+5.
\]
The block $Z$ uses $s=2\ell w$ source bits for the polynomial coefficients. Since $w<2\max\{d_0,d\}$,
\[
 s=2\ell w<4\ell\max\{d_0,d\}
 =O\!\left(a+\ell[\log a+\log(1/\eta)+\log\ell]\right).
\]
The total output loss is $N-m=s+(n-m)$, where $n-m=O(a+\log(1/\eta))$ is the loss of $\Ext_0$. This proves~\eqref{eq:moment-output}. The same bound holds when $N-s<\max\{1,a\}$: then $N<s+\max\{1,a\}$, so choosing the absolute constant in~\eqref{eq:moment-output} sufficiently large makes its right-hand side nonpositive.

Computing $\Ext((x,z);j)$ requires one polynomial evaluation and one call to $\Ext_0$. Including the deterministic construction of the finite field~\cite{Shoup90}, the running time is polynomial in $N+\log(1/\eta)+\ell$.
\end{proof}

\begin{proof}[Proof of Theorem~\ref{thm:deficiency}]
Put $a=n-k$. For $0<\eta<1/4$, apply Theorem~\ref{thm:moment-tradeoff} with $N=n$, entropy threshold $k$, $h=\log a+\log(1/\eta)$, and
\[
 \ell=\max\left\{1,\left\lceil\frac{a}{\alpha h}\right\rceil\right\}.
\]
We have $a/\ell\leq\alpha h$, $\ell h\leq a/\alpha+h$, and $\log\ell\leq\log a+O_\alpha(1)$. Substituting in~\eqref{eq:moment-seed} and~\eqref{eq:moment-output} gives
\[
 \begin{aligned}
 d&\leq(1+\alpha)\log a+(2+\alpha)\log(1/\eta)+O_\alpha(1),\\
 n-m&=O_\alpha\!\left(a+\log(1/\eta)+\ell h\right)
 =O_\alpha(a+\log(1/\eta)).
 \end{aligned}
\]
Since $\ell=O_\alpha(n)$, the extractor is computable in polynomial time. For $1/4\leq\eta\leq1/2$, use the construction at error $1/8$; the seed bound changes by only $O_\alpha(1)$.
\end{proof}

%% file: matrix_samplers.tex
\section{Strong Matrix Samplers}\label{sec:matrix-samplers}

\subsection{From extractors to samplers}\label{sec:sample-tradeoff}
Fix $m,D\geq1$ and $0<\delta\leq\varepsilon\leq1/2$.
Put $\Lambda=\log(2D/\delta)$ and $a=\lceil\Lambda\rceil$.
For any integer $1\leq\ell\leq a$, apply
Theorem~\ref{thm:moment-tradeoff} at deficiency $a$,
error $\varepsilon/4$, and source length
\[
 N=m+\left\lceil
 C_0\bigl(\Lambda+\ell\log(\Lambda/\varepsilon)\bigr)
 \right\rceil,
\]
where $C_0$ is a sufficiently large absolute constant.
This gives a strong quantum-proof
$(N-a,\varepsilon/4)$ extractor
$F:\{0,1\}^N\times\{0,1\}^d\to\{0,1\}^m$.
By the correspondence~\eqref{eq:conditional-forward},
the map $\Samp(x,y)=F(x,y)$ is a strong $D$-dimensional
$(\delta,\varepsilon)$ matrix sampler, with
\begin{align}
 r&=N=m+O\bigl(\Lambda+\ell\log(\Lambda/\varepsilon)\bigr),
 \label{eq:sampler-tradeoff-r}\\
 t&=2^d=O\bigl(\varepsilon^{-2}\ell2^{\Lambda/\ell}\bigr).
 \label{eq:sampler-tradeoff-t}
\end{align}
Each sample is computable by one evaluation of $F$,
in time polynomial in $m+\Lambda+\log(1/\varepsilon)$.

\begin{proof}[Proof of Theorem~\ref{thm:extractor-sampling}]
Put $h=\log a+\log(1/\varepsilon)$ and take
\[
 \ell=\min\left\{
 a,\left\lceil\frac{a}{\alpha h}\right\rceil
 \right\}.
\]
Since $h\leq2a$, we have $\ell=O_\alpha(a/h)$ and
$\Lambda/\ell\leq\max\{1,\alpha h\}$. Thus
\[
 \ell=O_\alpha\!\left(
 \frac{\Lambda}{\log(\Lambda/\varepsilon)}
 \right),
 \qquad
 2^{\Lambda/\ell}
 =O_\alpha(\Lambda^\alpha\varepsilon^{-\alpha}).
\]
Substituting into~\eqref{eq:sampler-tradeoff-r}--\eqref{eq:sampler-tradeoff-t}
gives~\eqref{eq:main-sampler}.
\end{proof}

\begin{corollary}\label{cor:quadratic-accuracy}
For every $m,D\geq1$ and $0<\delta\leq\varepsilon\leq1/2$,
put $\Lambda=\log(2D/\delta)$.
There is an explicit strong $D$-dimensional $(\delta,\varepsilon)$
matrix sampler
$\Samp:\{0,1\}^r\times[t]\to\{0,1\}^m$ such that
\[
 r=m+O\bigl(\Lambda\log(\Lambda/\varepsilon)\bigr),
 \qquad
 t=O(\varepsilon^{-2}\Lambda).
\]
Each sample $\Samp(x,i)$ is computable in time polynomial in
$m+\Lambda+\log(1/\varepsilon)$.
\end{corollary}
\begin{proof}
Take $\ell=\lceil\Lambda\rceil$ in
\eqref{eq:sampler-tradeoff-r}--\eqref{eq:sampler-tradeoff-t}.
Then $\ell=O(\Lambda)$ and $2^{\Lambda/\ell}\leq2$,
giving the stated bounds.
\end{proof}

% The choice $\ell=\lceil\Lambda\rceil$ in
% \eqref{eq:sampler-tradeoff-r}--\eqref{eq:sampler-tradeoff-t}
% gives $t=O(\Lambda/\varepsilon^2)$ and
% $r=m+O(\Lambda\log(\Lambda/\varepsilon))$.
% These parameters also follow from the sampler composition and
% subsampling bounds of Xun and Zuckerman~\cite[Theorem~3,
% Lemmas~23--24, and Remark~9]{XZ}.
\subsection{Lower bounds}
\begin{lemma}\label{prop:strong-randomness-lower}
Let $m,D\geq1$ and $0<\delta,\varepsilon<1$. Every strong $D$-dimensional $(\delta,\varepsilon)$ matrix sampler $\Samp:\{0,1\}^r\times[t]\to\{0,1\}^m$ satisfies
\[
 r\geq m+\log D+\log(1/\delta)-1.
\]
If $0<\delta\leq\varepsilon\leq1/2$, then
\[
 t=\Omega\!\left(\varepsilon^{-2}\log(D/\delta)\right).
\]
\end{lemma}
\begin{proof}Put $M=2^r$ and $q=2^{m-1}$.
Partition a set of $\min\{M,Dq\}$ seeds into $D$ sets
$B_1,\ldots,B_D$, each of size
$\lfloor\min\{M/D,q\}\rfloor$ or $\lceil\min\{M/D,q\}\rceil$. For each  $i\in[t]$, define a diagonal test $f_i$ by
\[
 A_{ij}=\{\Samp(x,i):x\in B_j\},\qquad
 (f_i(z))_{jj}=2\mathbf1_{A_{ij}}(z)-1\quad(j\in[D]).
\]
Each $f_i(z)$ has norm one. Since $|A_{ij}|\leq|B_j|\leq2^{m-1}$, every seed $x\in B_j$ satisfies
\[
 \bigl(f_i(\Samp(x,i))-\E_{Z\sim U_m}f_i(Z)\bigr)_{jj}
 =2\left(1-\frac{|A_{ij}|}{2^m}\right)\geq1.
\]
Thus, for every $x\in\bigcup_{j=1}^D B_j$,
\[
 \left\|
 \frac1t\sum_{i=1}^t
 \left(f_i(\Samp(x,i))-\mathbb E_{Z\sim U_m}f_i(Z)\right)
 \right\|
 \geq 1.
\]
Since $\varepsilon<1$, every seed in $\bigcup_j B_j$ causes
the sampling error to exceed $\varepsilon$. Hence
\[
\begin{aligned}
 \min\left\{1,\frac{Dq}{M}\right\}
 &=\frac{|\bigcup_j B_j|}{M}\\
 &\leq
 \Pr_{X\sim U_r}\left[
  \left\|
   \frac1t\sum_{i=1}^t
   \left(f_i(\Samp(X,i))-\mathbb E_{Z\sim U_m}f_i(Z)\right)
  \right\|>\varepsilon
 \right]\\
 &\leq\delta,
\end{aligned}
\]
where the last inequality follows from the sampler guarantee.
Since $\delta<1$, we obtain $Dq/M\leq\delta$.
Thus $M\geq Dq/\delta$. Taking logarithms gives $r\geq m+\log D+\log(1/\delta)-1$.

\emph{Sample complexity.}
Assume $0<\delta\leq\varepsilon\leq1/2$.
Let $S_s$ be a sum of $s$ independent uniform signs, and put
$p_t=\Pr[|S_t|>\varepsilon t]$.
We first show that, for an absolute constant $C$,
\begin{equation}\label{eq:sign-tail-lower}
 p_t\leq\varepsilon
 \quad\Longrightarrow\quad
 \log(1/p_t)\leq C\varepsilon^2t.
\end{equation}

Suppose $p_t\leq\varepsilon$.
The moment bounds $\E S_s^2=s$ and $\E S_s^4\leq3s^2$
give, by Paley--Zygmund and symmetry,
\begin{equation}\label{eq:sign-block-tail}
 \Pr[S_s>\varepsilon s]
 =\frac12\Pr[|S_s|>\varepsilon s]
 \geq\frac3{32}
 \qquad\text{if }\varepsilon^2s\leq\frac14.
\end{equation}
If $\varepsilon^2t<1/64$, then $\varepsilon<1/8$ and
$p_t\geq3/16$, contradicting $p_t\leq\varepsilon$.
Hence $\varepsilon^2t\geq1/64$.

Put $\ell=\lfloor1/(4\varepsilon^2)\rfloor$ and partition
the $t$ signs into
$v=\lceil t/\ell\rceil\leq8\varepsilon^2t+1=O(\varepsilon^2t)$
blocks of size at most $\ell$.
Requiring every block sum to exceed $\varepsilon$ times its
size gives $p_t\geq(3/32)^v$, by independence
and~\eqref{eq:sign-block-tail}.
Therefore
\[
 \log(1/p_t)\leq v\log(32/3)=O(\varepsilon^2t).
\]

We now apply~\eqref{eq:sign-tail-lower} to the sampler.
Choose independent uniform signs $\xi_{ij}\in\{-1,1\}$
for $i\in[t]$ and $j\in[D]$.
For $z=(z_1,\ldots,z_m)$, define
\[
 g_i(z)=(-1)^{z_1}
 \operatorname{diag}(\xi_{i1},\ldots,\xi_{iD}).
\]
Every $g_i$ satisfies $\norm{g_i(z)}=1$ and
$\E_{Z\sim U_m}g_i(Z)=0$.
For each fixed sampler seed $x$, the signs
$(-1)^{\Samp(x,i)_1}\xi_{ij}$ remain independent and uniform.
The $D$ diagonal entries of
$t^{-1}\sum_i g_i(\Samp(x,i))$ are independent copies of $S_t/t$.

For each fixed $x$, the probability over $\xi$ that at least one
diagonal entry exceeds $\varepsilon$ in absolute value is
$1-(1-p_t)^D$.
For each fixed $\xi$, the sampler guarantee bounds the failure
probability over $X\sim U_r$ by $\delta$.
Averaging over $X$ and $\xi$ gives
\begin{equation}\label{eq:strong-sampler-signs}
 1-(1-p_t)^D\leq\delta,
 \qquad
 p_t\leq\frac{-\ln(1-\delta)}{D}
 \leq\frac{2\delta}{D}.
\end{equation}
The first inequality also gives
$p_t\leq\delta\leq\varepsilon\leq1/2$.
Combining~\eqref{eq:strong-sampler-signs}
and~\eqref{eq:sign-tail-lower}, we obtain
\[
 \log(D/\delta)
 \leq1+\log(1/p_t)
 \leq2\log(1/p_t)
 \leq2C\varepsilon^2t.
\]
Therefore $t=\Omega(\varepsilon^{-2}\log(D/\delta))$.
\end{proof}

%For ordinary samplers, enumerating the entire domain gives zero error without using any randomness.

%% file: ack.tex
\section{Acknowledgement}

We thank David Zuckerman, Salil Vadhan, and Zhiyang Xun for helpful discussions on this topic.

\paragraph{Statement on the Use of AI}
%We initiated the topics and questions of this work without using AI. The question came to us when we survey the two interesting lines of work, i.e. quantum-proof extractors and matrix samplers. We would like to find the corresponding sampler for quantum-proof extractors and the corresponding extractor for samplers. We also had some general ideas of improving those existing constructions before using AI. Then we discussed with GPT and used it substantially to develop the solutions and proofs as well as to do proof checking. The write-up is done by us humans, though we also use GPT to polish sentences. We are responsible for the correctness. 

 We formulated the research topics and questions of this work without using AI. The questions arose while we were surveying two lines of work: quantum-proof extractors and matrix samplers. We sought to identify the samplers corresponding to quantum-proof extractors and the extractors corresponding to matrix samplers. We also had some general ideas for improving the existing constructions before using AI. We used ChatGPT (GPT-6 Astra) to assist with developing the solutions and proofs and with proof checking, making substantial use of the tool for these tasks. The tool materially affected the development of the solutions and proofs. The manuscript was written by the authors, who also used the tool to polish sentences. The authors verified the correctness and originality of all content, including references.

%% file: refer.tex
\par
\bibliographystyle{alpha}
\bibliography{refer}

%% file: appendices/smooth_min_entropy.tex
\section{Proofs for Smooth Min-Entropy}\label{app:smooth-proofs}

\begin{proof}[Proof of Lemma~\ref{lem:smooth-classical-leakage}]
Put $k=\hmin^\tau(X\mid E)_\rho$. By~\eqref{eq:smooth-min-entropy}, choose a subnormalized cq state $\widetilde\rho_{XE}$ such that
\[
 P(\widetilde\rho_{XE},\rho_{XE})\leq\tau,
 \qquad \hmin(X\mid E)_{\widetilde\rho}=k.
\]
By the extension property of purified distance~\cite[Corollary~9]{TCR},
there is a subnormalized extension $\widetilde\rho_{XLE}$ of
$\widetilde\rho_{XE}$ satisfying
\[
 \Tr_L\widetilde\rho_{XLE}=\widetilde\rho_{XE},
 \qquad
 P(\widetilde\rho_{XLE},\rho_{XLE})
 =P(\widetilde\rho_{XE},\rho_{XE})\leq\tau.
\]
Delete the off-diagonal blocks in the classical bases of $X,L$. Dephasing preserves the marginal $\widetilde\rho_{XE}$ and leaves $\rho_{XLE}$ unchanged. Contractivity gives $P(\widetilde\rho_{XLE},\rho_{XLE})\leq\tau$ after dephasing.

The classical leakage inequality~\eqref{eq:chain}, stated in~\cite[Lemma~B.3]{DPVR}, also holds for subnormalized states. Normalizing a state of trace $p$ adds $\log p$ to both conditional min-entropies. Applying~\eqref{eq:chain} to $\widetilde\rho_{XLE}$ gives
\begin{align*}
 \hmin^\tau(X\mid LE)_\rho
 &\geq\hmin(X\mid LE)_{\widetilde\rho}\\
 &\geq\hmin(X\mid E)_{\widetilde\rho}-\log|\mathcal L|\\
 &=k-\log|\mathcal L|.
\end{align*}
\end{proof}

\begin{proof}[Proof of Lemma~\ref{lem:smooth-extraction}]
For $\tau>1/4$, the error bound in Lemma~\ref{lem:smooth-extraction} follows from the upper bound of one on trace distance. Assume $0\leq\tau\leq1/4$. Choose a subnormalized cq state $\widetilde\rho_{XE}$ with $P(\widetilde\rho_{XE},\rho_{XE})\leq\tau$ and $\hmin(X\mid E)_{\widetilde\rho}\geq k+1$. Let $p=\Tr\widetilde\rho_{XE}$. Then $p\geq1-\tau^2\geq15/16$. The normalized state $\widehat\rho_{XE}=\widetilde\rho_{XE}/p$ therefore satisfies
\[
 \hmin(X\mid E)_{\widehat\rho}
 =\hmin(X\mid E)_{\widetilde\rho}+\log p\geq k,
 \qquad
 \dist(\widehat\rho_{XE},\rho_{XE})
 \leq\tau+\frac{1-p}{2}\leq2\tau.
\]
Apply $\Ext$ to $\widehat\rho_{XE}$ with the independent seed $Y$. As in the smoothing argument of~\cite[Lemma~3.5]{DPVR}, contractivity and the triangle inequality give
\begin{align*}
 &\dist\!\left(\rho_{\Ext(X,Y)YE},
       \tau_m\otimes\tau_d\otimes\rho_E\right)\\
 &\qquad\leq
 \dist\!\left(\widehat\rho_{\Ext(X,Y)YE},
       \tau_m\otimes\tau_d\otimes\widehat\rho_E\right)
 +2\dist(\widehat\rho_{XE},\rho_{XE})
 \leq\eta+4\tau.
\end{align*}
\end{proof}

%% file: appendices/sampling_extraction_tools.tex
\section{Derivations of the Sampling and Extraction Tools}\label{app:sampling-extraction-tools}

\subsection{Entropy sampling with a public seed}
\begin{theorem}[Vadhan's sampler; {\cite[Lemma~8.4]{Vad04}}]\label{thm:vadhan-distinct-sampler}
There is an absolute constant $C>0$ with the following property. Fix an integer $N\geq1$ and parameters $0<\varepsilon<\beta<1$ and $0<\delta\leq1/2$. For every integer $C\varepsilon^{-2}\log(1/\delta)\leq t\leq N$, there is an explicit map
\[
 S:\{0,1\}^{r}\to\binom{[N]}t,
 \qquad r\leq\log(N/t)+\poly(1/\varepsilon)\log(1/\delta).
\]
For every $f:[N]\to[0,1]$ with $N^{-1}\sum_{j=1}^N f(j)\geq\beta$, a uniform seed $R\sim U_r$ satisfies
\[
 \Pr_R\!\left[\frac1t\sum_{j\in S(R)}f(j)<\beta-\varepsilon\right]\leq\delta.
\]
\end{theorem}

\begin{corollary}[Distinct samples]\label{cor:distinct-samples}
Fix $0<\varepsilon<1$. There is a constant $C_\varepsilon>0$ with the following property. Let $N\geq1$ be an integer and $0<\delta\leq1/2$. For every integer sample size $C_\varepsilon\log(1/\delta)\leq t\leq N$, there is an explicit map $S:\{0,1\}^{r}\to\binom{[N]}t$. The seed length is $r=O_\varepsilon(\log N+\log(1/\delta))$. For a uniform seed $R\sim U_r$ and every $f:[N]\to[0,1]$,
\begin{equation}\label{eq:distinct-sampling-guarantee}
 \Pr_R\!\left[\frac1t\sum_{j\in S(R)}f(j)
       \leq\frac1N\sum_{j=1}^N f(j)-\varepsilon\right]\leq\delta.
\end{equation}
\end{corollary}
\begin{proof}
Choose the sampler in Theorem~\ref{thm:vadhan-distinct-sampler} with threshold $\beta=1/2$ and accuracy $\varepsilon/4$. For any test $f$, put $\overline f=N^{-1}\sum_{j=1}^N f(j)$ and define $g(j)=(f(j)-\overline f+1)/2$. Then $0\leq g\leq1$, $N^{-1}\sum_jg(j)=1/2$, and
\begin{align*}
 \Pr_R\!\left[\frac1t\sum_{j\in S(R)}f(j)\leq\overline f-\varepsilon\right]
 &=\Pr_R\!\left[\frac1t\sum_{j\in S(R)}g(j)\leq\frac12-\frac\varepsilon2\right]\\
 &\leq\Pr_R\!\left[\frac1t\sum_{j\in S(R)}g(j)<\frac12-\frac\varepsilon4\right]
 \leq\delta.
\end{align*}
The sample-size and seed-length bounds in Theorem~\ref{thm:vadhan-distinct-sampler} give the stated parameters.
\end{proof}

The entropy-sampling theorem below uses trace-norm smoothing as defined in~\eqref{eq:trace-smooth-min-entropy}.

\begin{theorem}[K\"onig--Renner entropy sampling; {\cite[Corollary~6.19]{KR}}]\label{thm:quantum-entropy-sampling}
Let $N,w,t\geq1$ be integers with $t\leq N$. Fix $0<\varepsilon<1$ and $0<\delta\leq1/2$. Let $\rho_{XE}$ be a normalized cq state with $X=(X_1,\ldots,X_N)\in(\{0,1\}^w)^N$. Let $J$ be a random $t$-element subset of $[N]$, independent of $XE$. Suppose
\[
 \Pr_J\!\left[\frac1t\sum_{j\in J}f(j)
       \leq\frac1N\sum_{j=1}^N f(j)-\varepsilon\right]\leq\delta
 \qquad\text{for every }f:[N]\to[0,1].
\]
Put $B=X_J$, with the sampled symbols in increasing coordinate order. Set
\[
 \kappa=\frac{N}{tw},\qquad s=2\cdot2^{-\varepsilon Nw}+3\delta^{1/4}.
\]
If $\kappa\leq0.15$, then
\begin{equation}\label{eq:kr-entropy-sampling}
 \frac{\hmin^{s,\|\cdot\|_1}(B\mid JE)_\rho}{wt}
 \geq\frac{\hmin(X\mid E)_\rho}{Nw}
             -3\varepsilon-2\kappa\log(1/\kappa).
\end{equation}
\end{theorem}

\begin{proof}[Proof of Lemma~\ref{lem:public-entropy-sampling}]
Choose $S$ from Corollary~\ref{cor:distinct-samples} and put $J=S(R)$ and $s=2\cdot2^{-\varepsilon Nw}+3\delta^{1/4}$. To replace $J$ by $R$ in the conditioning register, define the classical channel
\[
 \Phi(\ket j\bra j)
 =\sum_{r:S(r)=j}\Pr[R=r\mid J=j]\ket r\bra r
 \qquad(j\in\supp(J)).
\]
Since $R$ is independent of $XE$ and $B=X_J$, we have
\[
 \rho_{BE\mid R=r}=\rho_{BE\mid J=S(r)},
 \qquad \Phi_{J\to R}(\rho_{BJE})=\rho_{BRE}.
\]
Applying $\Phi$ to the conditioning register cannot decrease the smooth min-entropy of $B$: trace distance is contractive, and conditional min-entropy satisfies data processing. Since $2s\leq\tau^2$, equation~\eqref{eq:smoothing-conversion} and Theorem~\ref{thm:quantum-entropy-sampling} give
\begin{align*}
 \hmin^\tau(B\mid RE)_\rho
 &\geq\hmin^{s,\|\cdot\|_1}(B\mid RE)_\rho\\
 &\geq\hmin^{s,\|\cdot\|_1}(B\mid JE)_\rho\\
 &\geq\bigl(\mu-3\varepsilon-2\kappa\log(1/\kappa)\bigr)wt.
\end{align*}
\end{proof}

\subsection{Short-seed extraction}
\begin{theorem}[Short-seed hashing; {\cite[Theorem~10]{LHL}}]\label{thm:tssr-short-seed}
For every pair of integers $1\leq m\leq n$ and every $0<\eta_0\leq1/2$, there is an explicit map
\[
 H:\{0,1\}^n\times\{0,1\}^d\to\{0,1\}^m,
 \qquad d\leq2m+O(\log(n/\eta_0)).
\]
For every normalized cq state $\rho_{XE}$ with $X\in\{0,1\}^n$ and an independent uniform seed $W\sim U_d$, put $Z=H(X,W)$. Then
\begin{equation}\label{eq:tssr-short-seed}
 \min_{\sigma_{WE}}
 \dist\!\left(\rho_{ZWE},\tau_m\otimes\sigma_{WE}\right)
 \leq3\eta_0+\frac12
 \sqrt{2^{m-\hmin^{\eta_0}(X\mid E)_\rho}(2/\eta_0^2+1)},
\end{equation}
where the minimum ranges over normalized states on $WE$.
\end{theorem}

\begin{proof}[Proof of Lemma~\ref{lem:short-hash}]
Apply Theorem~\ref{thm:tssr-short-seed} with $\eta_0=\eta/16$. Fix a normalized cq source $\rho_{XE}$ and write $k=\hmin(X\mid E)_\rho$. Take an independent uniform seed $W\sim U_d$ and put $Z=H(X,W)$. Let $\sigma_{WE}$ attain the minimum in~\eqref{eq:tssr-short-seed}. Tracing out $Z$ gives
\[
 \dist(\rho_{WE},\sigma_{WE})
 \leq\dist(\rho_{ZWE},\tau_m\otimes\sigma_{WE}).
\]
The triangle inequality gives
\begin{align*}
 \dist(\rho_{ZWE},\tau_m\otimes\rho_{WE})
 &\leq\dist(\rho_{ZWE},\tau_m\otimes\sigma_{WE})
       +\dist(\tau_m\otimes\sigma_{WE},\tau_m\otimes\rho_{WE})\\
 &=\dist(\rho_{ZWE},\tau_m\otimes\sigma_{WE})
       +\dist(\sigma_{WE},\rho_{WE})\\
 &\leq2\dist(\rho_{ZWE},\tau_m\otimes\sigma_{WE}).
\end{align*}
Since $\hmin^{\eta_0}(X\mid E)_\rho\geq k$, equation~\eqref{eq:tssr-short-seed} now yields
\[
 \dist(\rho_{ZWE},\tau_m\otimes\rho_{WE})
 \leq6\eta_0+\sqrt{2^{m-k}(2/\eta_0^2+1)}\leq\eta
\]
whenever $k\geq m+4\log(1/\eta)+12$. The seed $W$ is uniform and independent of $XE$, so $\rho_{WE}=\tau_d\otimes\rho_E$, as required for strong extraction. The seed bound is $d\leq2m+O(\log(n/\eta))$.

The hash family in Theorem~\ref{thm:tssr-short-seed} consists of polynomial evaluation followed by two-universal hashing. Polynomial evaluation and two-universal hashing run in polynomial time. The required finite field can be constructed in fixed characteristic using~\cite{Shoup90}.
\end{proof}